\documentclass[10pt,reqno]{article}
\usepackage{fullpage} 
\usepackage{amssymb} 
\usepackage{graphicx} 
\usepackage{amsmath} 
\usepackage{amsthm,amssymb,latexsym} 
\usepackage{amstext, amsfonts,amsbsy} 
\usepackage{enumerate} 
\usepackage{upgreek} 
\usepackage{float} 
\usepackage{color} 
\usepackage{accents} 
\usepackage{mathtools} 
\usepackage{tikz}
\usetikzlibrary{matrix,graphs,arrows,positioning,calc,decorations.markings,shapes.symbols,shapes.geometric,shapes.misc}
\usepackage{calc, ifthen, ae, xstring, etoolbox, xkeyval,pgfplots}
\usepackage[pdftex,bookmarks,colorlinks,breaklinks]{hyperref}  
\definecolor{dullmagenta}{rgb}{0.4,0,0.4}   
\definecolor{darkblue}{rgb}{0,0,0.4}
\hypersetup{linkcolor=red,citecolor=blue,filecolor=dullmagenta,urlcolor=darkblue} 
\usepackage{eucal}

\newtheorem{theorem}{Theorem}

\newtheorem{lemma}[theorem]{Lemma}
\newtheorem{definition}[theorem]{Definition}

\newcommand{\Pain}[1]{\textup{P}_{\mathrm{#1}}}

\theoremstyle{definition}

\theoremstyle{remark}
\newtheorem{remark}{Remark}

\newtheorem*{notation*}{Notation} 

\numberwithin{equation}{section}

\begin{document}

{\noindent\Large\bf What is the symmetry group of a d-$\Pain{II}$ discrete Painlev\'e equation? 
}
\medskip

\begin{flushleft}

\textbf{Anton Dzhamay}\\
Beijing Institute of Mathematical Sciences and Applications (BIMSA)\\ 
No.~544, Hefangkou Village Huaibei Town, Huairou District, Beijing 101408 \\
E-mail: \href{mailto:adzham@bimsa.cn}{\texttt{adzham@bimsa.cn}}\qquad 
ORCID ID: \href{https://orcid.org/0000-0001-8400-2406}{\texttt{0000-0001-8400-2406}}\\[5pt]

\textbf{Yang Shi}\\
College of Science and Engineering, Flinders University,\\ Flinders at Tonsley, South Australia 5042, Australia\\
E-mail: \href{mailto:yshi7200@gmail.com}{\texttt{yshi7200@gmail.com}}\qquad
ORCID ID: \href{https://orcid.org/0000-0001-5862-1406}{\texttt{0000-0001-5862-1406}}\\[8pt]

\textbf{Alexander Stokes}\\
Waseda Institute for Advanced Study, Waseda University (WIAS) \\1-21-1 Nishi Waseda, Shinjuku-ku, Tokyo 169-0051, Japan\\
E-mail: \href{mailto:stokes@aoni.waseda.jp}{\texttt{stokes@aoni.waseda.jp}}\qquad
ORCID ID: \href{https://orcid.org/0000-0001-6874-7141}{\texttt{0000-0001-6874-7141}}\\[8pt]

\textbf{Ralph Willox}\\
Graduate School of Mathematical Sciences, The University of Tokyo,\\ 3--8--1 Komaba, Meguro-ku, \\Tokyo 153-8914, Japan\\
E-mail: \href{mailto:willox@ms.u-tokyo.ac.jp}{\texttt{willox@ms.u-tokyo.ac.jp}}\qquad
ORCID ID: \href{https://orcid.org/0000-0002-8054-0239}{\texttt{0000-0002-8054-0239}}\\[8pt]



\end{flushleft}

\date{}

\begin{abstract}
	The symmetry group of a (discrete) Painlev\'e equation provides crucial information on the properties of the equation. In this paper we argue against the commonly-held belief that the symmetry group of a given equation is solely determined by its surface type as given in the famous Sakai classification. We will dispel this misconception on a specific example of a d-$\Pain{II}$ equation which corresponds to a half-translation on the weight lattice of the root system dual to its surface-type root lattice, but which becomes a genuine translation on a sub-lattice thereof that corresponds to its real symmetry group. The latter fact is shown in two different ways: first by a brute force calculation and second through the use of normalizer theory, which we believe to be an extremely useful tool for this purpose. We finish the paper with the analysis of a sub-case of our main example which arises in the study of gap probabilities for Freud unitary ensembles, and the symmetry group of which is even further restricted due to the appearance of a nodal curve on the surface on which the equation is regularized.
\end{abstract}


\section{Introduction} 
\label{sec:introduction}

The theory of discrete Painlev\'e equations is roughly a quarter century old and during this time we have 
achieved a very good understanding of these equations. 
However, a lot of terminology reflects the historic
development of the theory, which can lead to various misconceptions. One of the goals of the present paper
is to address some of these misconceptions, in particular the relationship between the \emph{surface-type}
and the \emph{symmetry-type} classification schemes, as well as the relationship between an equation and 
its symmetry group. But let us begin with the following definition.

\begin{definition}
	An (abstract) discrete Painlev\'e equation is a triple $(\mathcal{R}, \mathcal{R}^{\perp}, [w])$, where 
	$\mathcal{R}$ and $\mathcal{R}^{\perp}$ are two root sub-systems 
	(described by affine Dynkin diagrams) of the affine root system $E_{8}^{(1)}$. The root system
	$\mathcal{R}$ describes the geometry of the configuration space of the dynamics, 
	the (fully) extended affine Weyl group $\widehat{W}$ (see Definition~\ref{def:fully-extended-W}) of type $\mathcal{R}^{\perp}$ describes the symmetry group of this
	configuration space, and the equation itself is described by an element $w\in \widehat{W}(\mathcal{R}^{\perp})$ of infinite order 	
	(a \emph{translation} or \emph{quasi-translation}), where $[w]$ is its equivalence class w.r.t.~conjugations in $\widehat{W}(\mathcal{R}^{\perp})$. 
\end{definition}

This definition may seem rather unconventional. In particular, there is no actual equation in the definition.
We think that this may, in fact, be the benefit of the suggested approach. To see this, as well as to understand
the connection to a more traditional definition of a discrete Painlev\'e equation, we need to revisit some history.

The name \emph{discrete Painlev\'e equation} is of course due to connections with \emph{differential Painlev\'e equations}. Recall that  
Painlev\'e equations appeared at the beginning of the $\text{XX}^{\text{th}}$ century
in an attempt to define \emph{nonlinear special functions} as solutions of
\emph{nonlinear} ordinary differential equations (ODEs), similar to the classical special functions such as Airy, Bessel, and many others, which solve linear differential equations.
For solutions of nonlinear ODE, one encounters the phenomenon of movable singular points (i.e., depending on the initial conditions), and if such a singularity is for example a branch point, we cannot talk about the Riemann surface for the general solution of the equation.
P. Painlev\'e suggested to study ODEs whose general solution have no movable critical points other than poles. We now say that
an ODE has the \emph{Painlevé Property} if the general solution of the equation is free of movable
critical points where it loses local single-valuedness. Painlev\'e, together with his student B. Gambier, found 
that an equation in the form $y'' = R(y',y,t)$ that has the Painlev\'e property can be put into one of fifty \emph{canonical forms}, 
of which \emph{six} can not be reduced to linear equations or solved in terms of classical special functions. These equations are now 
known as \emph{Painlev\'e equations} and their solutions are called \emph{Painlev\'e transcendents}. Differential Painlev\'e equations 
play an increasingly important role in modern Mathematical Physics, and we recommend \cite{Con:1999:PP,ConMus:2020:PH,FokItsKapNov:2006:PT} 
and references therein to the interested reader.

The term \emph{discrete Painlev\'e equation} probably appeared in the literature for the first time in the paper \cite{ItcKitFok:1991:MMTQGISPDE},
following earlier works \cite{BreKaz:1990:ESFTCS,GroMig:1990:NTTQG} on two-dimensional quantum gravity. This result quickly attracted
the attention of researchers working with discrete integrable systems, \cite{NijPap:1991:SRILDAPE,PapNijGraRam:1992:IDPDAPE}. A large number of
examples of discrete Painlev\'e equation has been obtained in the work of B.~Grammaticos and A.~Ramani who applied
the singularity confinement criterion to deautonomizations of discrete integrable autonomous mappings such as the QRT maps, see the survey
\cite{GraRam:2004:DPER} and references therein. In this approach, the term \emph{discrete Painlev\'e equation} denoted
a certain second-order non-autonomous recurrence relation that has one of the differential Painlev\'e equations as a
continuous limit. It is worth noting that the first example of such a non-autonomous recurrence goes back at least
50 years earlier, to the 1939 paper of J.~Shohat on orthogonal polynomials \cite{Sho:1939:DEOP}, but Shohat did not take a continuous limit and so the relationship to
differential Painlev\'e equations was missed.

This approach resulted in names such as d-$\Pain{II}$, or $q$-$\Pain{VI}$, or alt.~d-$\Pain{I}$ given to certain recurrences, based on the existence 
of a continuous limit. However, it quickly became clear that this naming scheme is quite confusing, and also that there are a lot more discrete Painlev\'e 
equations than the differential ones. An important breakthrough is due to H. Sakai \cite{Sak:2001:RSAWARSGPE} who, following the earlier works of 
K. Okamoto \cite{Oka:1979:FAESOPCFPP} for the differential Painlev\'e equations, approached discrete Painlev\'e equations from the point of view of 
algebraic geometry. Sakai's work clarified the algebraic nature of discrete Painlev\'e equations, and also resulted in a clear classification scheme. 
However, in using this classification scheme certain care is necessary and this is precisely the point that we want to address.

Let us first sketch some of the key points of Sakai's approach, referring the reader to the survey \cite{KajNouYam:2017:GAPE}, as well as Sakai's original paper
\cite{Sak:2001:RSAWARSGPE}, for careful statements and details. Sakai's point of view is that discrete Painlev\'e equations are discrete dynamical
systems on certain families of rational algebraic surfaces, called generalized Halphen
surfaces. 
These families can be obtained by blowing up a configuration of eight points on $\mathbb{P}^{(1)}_{\mathbb{C}} \times \mathbb{P}^{1}_{\mathbb{C}}$
(or, alternatively, a configuration of $9$ points on $\mathbb{P}^{2}_{\mathbb{C}}$). In the generic
case these points lie on a unique elliptic curve $D$ that can be thought of as a divisor of a section of the anti-canonical bundle, i.e., the polar divisor of
some rational 2-form $\omega$, $[D] = -[(\omega)] = -\mathcal{K}_{\mathbb{P}^{1} \times \mathbb{P}^{1}}$. 
After blowing up those points we obtain a surface $X$ with unique effective anti-canonical divisor $-K_{X}$. The Picard lattice of $X$ is generated by the 
coordinate classes $\mathcal{H}_{i}$ and the classes $\mathcal{E}_{i}$ of the exceptional divisors of the blowups,
\begin{equation*}
	\operatorname{Pic}(X) = \operatorname{Span}_{\mathbb{Z}}\{\mathcal{H}_{1},\mathcal{H}_{2}, \mathcal{E}_{1},\ldots,\mathcal{E}_{8}\},\qquad 
	- \mathcal{K}_{X} = [-K_{X}]= 2\mathcal{H}_{1} + 2 \mathcal{H}_{2} - \mathcal{E}_{1} - \cdots - \mathcal{E}_{8}.
\end{equation*}
Varying locations of the blowup points, but still keeping them in a general position, creates a \emph{family} $\mathcal{X}$ of such surfaces; the type of this family is denoted
by the symbol $A_{0}^{(1)}$. The group of \emph{symmetries} of this family, in the sense of \emph{Cremona isometries} on the level of the Picard lattice and their realisation as the \emph{Cremona action} by automorphisms of the family,
is the affine Weyl group $W\left(E^{(1)}_{8}\right)$. This group has \emph{translation} elements that, when acting on the surface family, 
define \emph{elliptic} discrete Painlev\'e equations. The location of points in the configuration evolves with each step, so the dynamics is non-autonomous. The embedding of the curve $D\subset \mathbb{P}^{1} \times \mathbb{P}^{1}$ can be given in terms of elliptic functions, and the point evolution becomes additive in the argument
of these elliptic functions. 
When written in coordinates, say, in the affine chart of $\mathbb{P}^{1} \times \mathbb{P}^{1}$, the coordinates of
blowup points become coefficients in the evolution equations, and this is what is meant by an elliptic difference equation.

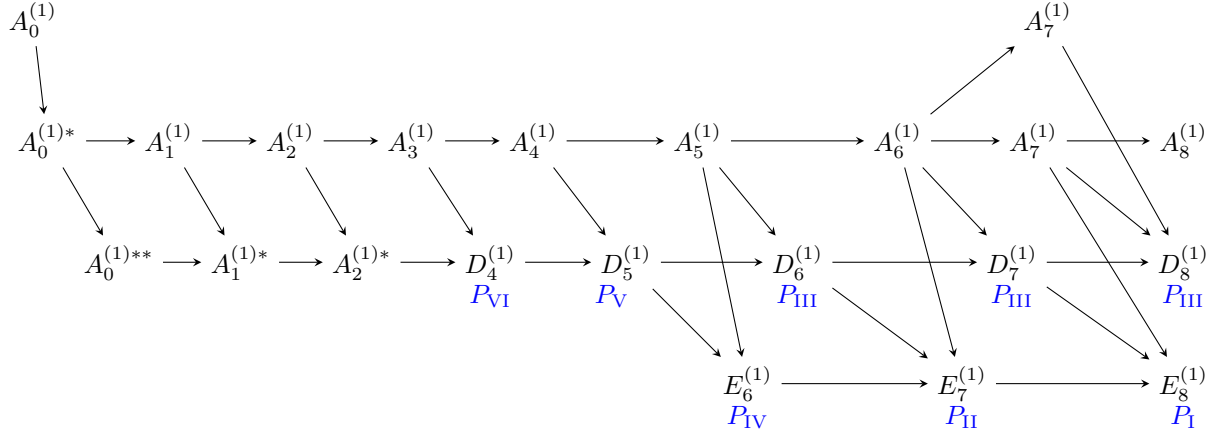
\begin{figure}[ht]
{
			\begin{tikzpicture}[>=stealth,scale=0.95]
			\node (e8e) at (0.8,3.4) {$A_{0}^{(1)}$};
			\node (a1qa) at (15,3.4) {$A_{7}^{(1)}$};
			\node (e8q) at (1,1.7) {$A_{0}^{(1)*}$};
			\node (e7q) at (2.7,1.7) {$A_{1}^{(1)}$};	
			\node (e6q) at (4.4,1.7) {$A_{2}^{(1)}$};	
			\node (d5q) at (6.1,1.7) {$A_{3}^{(1)}$};	
			\node (a4q) at (7.8,1.7) {$A_{4}^{(1)}$};	
			\node (a21q) at (10.1,1.7) {$A_{5}^{(1)}$};	
			\node (a11q) at (12.9,1.7) {$A_{6}^{(1)}$};	
			\node (a1q) at (14.8,1.7) {$A_{7}^{(1)}$};	
			\node (a0q) at (16.9,1.7) {$A_{8}^{(1)}$};						
			\node (e8d) at (2,0) {$A_{0}^{(1)**}$};	
			\node (e7d) at (3.7,0) {$A_{1}^{(1)*}$};	
			\node (e6d) at (5.4,0) {$A_{2}^{(1)*}$};	
			\node (d4d) at (7.2,0) {$D_{4}^{(1)}$};	
			\node (a3d) at (9.1,0) {$D_{5}^{(1)}$};	
			\node (a11d) at (11.5,0) {$D_{6}^{(1)}$};	
			\node (a1d) at (14.5,0) {$D_{7}^{(1)}$};	
			\node (a0d) at (16.9,0) {$D_{8}^{(1)}$};	
			\node (a2d) at (10.8,-1.7) {$E_{6}^{(1)}$};	
			\node (a1da) at (13.8,-1.7) {$E_{7}^{(1)}$};	
			\node (a0da) at (16.9,-1.7) {$E_{8}^{(1)}$};	
			\draw[->] (e8e) -> (e8q);	\draw[->] (a1qa) -> (a0d); 	
			\draw[->] (e8q) -> (e7q); 	\draw[->] (e8q) -> (e8d); 	
			\draw[->] (e7q) -> (e6q); 	\draw[->] (e7q) -> (e7d); 	
			\draw[->] (e6q) -> (d5q); 	\draw[->] (e6q) -> (e6d); 	
			\draw[->] (d5q) -> (a4q); 	\draw[->] (d5q) -> (d4d); 	
			\draw[->] (a4q) -> (a21q); 	\draw[->] (a4q) -> (a3d); 	
			\draw[->] (a21q) -> (a11q); \draw[->] (a21q) -> (a11d); \draw[->] (a21q) -> (a2d);	
			\draw[->] (a11q) -> (a1q); 	\draw[->] (a11q) -> (a1d); 	\draw[->] (a11q) -> (a1qa); 	\draw[->] (a11q) -> (a1da);
			\draw[->] (a1q) -> (a0q); 	\draw[->] (a1q) -> (a0d);	\draw[->] (a1q) -> (a0da);
			\draw[->] (e8d) -> (e7d);
			\draw[->] (e7d) -> (e6d);
			\draw[->] (e6d) -> (d4d);
			\draw[->] (d4d) -> (a3d);	
			\draw[->] (a3d) -> (a11d);	\draw[->] (a3d) -> (a2d);	
			\draw[->] (a11d) -> (a1d);	\draw[->] (a11d) -> (a1da);	
			\draw[->] (a1d) -> (a0d);	\draw[->] (a1d) -> (a0da);	
			\draw[->] (a2d) -> (a1da);	\draw[->] (a1da) -> (a0da);
			\node [blue] at ($(d4d.south) + (0,-0.1)$) {$P_{\text{VI}}$};
			\node [blue] at ($(a3d.south) + (-0.2,-0.1)$) {$P_{\text{V}}$};
			\node [blue] at ($(a11d.south) + (0,-0.1)$) {$P_{\text{III}}$};
			\node [blue] at ($(a1d.south) + (0,-0.1)$) {$P_{\text{III}}$};
			\node [blue] at ($(a0d.south) + (0,-0.1)$) {$P_{\text{III}}$};
			\node [blue] at ($(a2d.south) + (0,-0.1)$) {$P_{\text{IV}}$};
			\node [blue] at ($(a1da.south) + (0,-0.1)$) {$P_{\text{II}}$};
			\node [blue] at ($(a0da.south) + (0,-0.1)$) {$P_{\text{I}}$};
		\end{tikzpicture}
} 
\caption{Surface-type classification scheme for Painlev\'e equations}
\label{fig:Surface-Type-Classification}
\end{figure}

The next step is to consider various degenerations, e.g., the elliptic curve on which our configuration of points lie can degenerate into a rational curve with a cusp or a node, and then under further degenerations become reducible, where each component is a rational curve. 
We can then write this decomposition of the anti-canonical divisor 
into irreducible components as $-K_{X} = \sum_{i=1}^{n} m_{i} D_{i}$, where $m_{i}\in \mathbb{Z}_{>0}$ are the multiplicities. The 
intersection configuration (w.r.t. to the usual intersection product on $\operatorname{Pic}(X)$ 
given on the generators as $\mathcal{H}_{1}\bullet \mathcal{H}_{2} = - \mathcal{E}_{i}^{2} = -1$ and zero otherwise) is given by the negative of a generalised Cartan matrix of affine type, and so the degenerations can be described using the language of affine Dynkin diagrams. The resulting classification scheme is given on Figure~\ref{fig:Surface-Type-Classification}.
Note that differential Painlev\'e equations also appear on this diagram via the types of surfaces which provide their spaces of initial conditions as constructed by Okamoto \cite{Oka:1979:FAESOPCFPP}.

\begin{figure}[ht]{\footnotesize
			\begin{tikzpicture}[>=stealth,scale=0.95]
				\node (e8e) at (0.8,3.4) {$\left(E_{8}^{(1)}\right)^{e}$};
				\node (a1qa) at (15,3.4) {$\left(A_{1}^{(1)}\right)^{q}_{|\alpha|^{2}=8}$};
				\node (e8q) at (1,1.7) {$\left(E_{8}^{(1)}\right)^{q}$};
				\node (e7q) at (2.7,1.7) {$\left(E_{7}^{(1)}\right)^{q}$};
				\node (e6q) at (4.4,1.7) {$\left(E_{6}^{(1)}\right)^{q}$};
				\node (d5q) at (6.1,1.7) {$\left(D_{5}^{(1)}\right)^{q}$};
				\node (a4q) at (7.8,1.7) {$\left(A_{4}^{(1)}\right)^{q}$};
				\node (a21q) at (10.1,1.7) {$\left((A_{2} + A_{1})^{(1)}\right)^{q}$};
				\node (a11q) at (12.9,1.7) {$\left((A_{1} + A_{1})^{(1)}\right)^{q}$};
				\node  at (12.2,1.25) {$|\alpha|^{2}=14$};
				\node (a1q) at (15.2,1.7) {$\left(A_{1}^{(1)}\right)^{q}$};
				\node  at (14.8,1.25) {$|\alpha|^{2}=4$};
				\node (a0q) at (16.9,1.7) {$\left(A_{0}^{(1)}\right)^{q}$};
				\node (e8d) at (2,0) {$\left(E_{8}^{(1)}\right)^{\delta}$};
				\node (e7d) at (3.7,0) {$\left(E_{7}^{(1)}\right)^{\delta}$};
				\node (e6d) at (5.4,0) {$\left(E_{6}^{(1)}\right)^{\delta}$};
				\node (d4d) at (7.2,0) {$\left(D_{4}^{(1)}\right)^{c,\delta}$};
				\node (a3d) at (9.1,0) {$\left(A_{3}^{(1)}\right)^{c,\delta}$};
				\node (a11d) at (11.5,0) {$\left(2A_{1}^{(1)}\right)^{c,\delta}$};
				\node (a1d) at (14.5,0) {$\left(A_{1}^{(1)}\right)^{c,\delta}$};
				\node (a0d) at (16.9,0) {$\left(A_{0}^{(1)}\right)^{c}$};
				\node (a2d) at (10.5,-1.7) {$\left(A_{2}^{(1)}\right)^{c,\delta}$};
				\node (a1da) at (13.3,-1.7) {$\left(A_{1}^{(1)}\right)^{c,\delta}$};
				\node (a0da) at (16.9,-1.7) {$\left(A_{0}^{(1)}\right)^{c}$};
				\draw[->] (e8e) -> (e8q);	\draw[->] (a1qa) -> (a0d); 	
				\draw[->] (e8q) -> (e7q); 	\draw[->] (e8q) -> (e8d); 	
				\draw[->] (e7q) -> (e6q); 	\draw[->] (e7q) -> (e7d); 	
				\draw[->] (e6q) -> (d5q); 	\draw[->] (e6q) -> (e6d); 	
				\draw[->] (d5q) -> (a4q); 	\draw[->] (d5q) -> (d4d); 	
				\draw[->] (a4q) -> (a21q); 	\draw[->] (a4q) -> (a3d); 	
				\draw[->] (a21q) -> (a11q); \draw[->] (a21q) -> (a11d); \draw[->] (a21q) -> (a2d);	
				\draw[->] (a11q) -> (a1q); 	\draw[->] (a11q) -> (a1d); 	\draw[->] (a11q) -> (a1qa); 	\draw[->] (a11q) -> (a1da);
				\draw[->] (a1q) -> (a0q); 	\draw[->] (a1q) -> (a0d);	\draw[->] (a1q) -> (a0da);
				\draw[->] (e8d) -> (e7d);
				\draw[->] (e7d) -> (e6d);
				\draw[->] (e6d) -> (d4d);
				\draw[->] (d4d) -> (a3d);	
				\draw[->] (a3d) -> (a11d);	\draw[->] (a3d) -> (a2d);	
				\draw[->] (a11d) -> (a1d);	\draw[->] (a11d) -> (a1da);	
				\draw[->] (a1d) -> (a0d);	\draw[->] (a1d) -> (a0da);	
				\draw[->] (a2d) -> (a1da);	\draw[->] (a1da) -> (a0da);	
			\end{tikzpicture}} 

\caption{Symmetry-type classification scheme for Painlev\'e equations}
\label{fig:Symmetry-Type-Classification}
\end{figure}

This diagram gives a \emph{complete} classification of the possible \emph{configuration spaces} on which discrete Painlev\'e dynamics
can occur. In \cite{Sak:2001:RSAWARSGPE} Sakai also computed the groups of Cremona isometries and their Cremona action for each of these families, as extensions of affine Weyl groups of types shown on
Figure~\ref{fig:Symmetry-Type-Classification}. Very often in the literature, especially in applications of discrete Painlev\'e 
equations, it is that second classification scheme that is being used. However, it has been known for a long time that there are 
examples of discrete Painlev\'e dynamics that stay on some proper sub-families of the general configuration space, and if we restrict our attention to such sub-families
the symmetry group of that sub-family is different from the full symmetry group \cite{Tak:2003:WGSTDQE,AtkHowJosNak:2016:GEDER,CarDzhTak:2017:FDIMDPE,Shi:2019:VAAATDPE}.
In this way, we can get symmetry groups that do not appear explicitly in the classification scheme on Figure~\ref{fig:Symmetry-Type-Classification}.
Two particularly important examples of such sub-families are related to the existence of so-called nodal curves, which form obstructions to Cremona isometries, as explained in \cite{Sak:2001:RSAWARSGPE}, and the notion of \emph{projective reduction} introduced 
by K.~Kajiwara, N.~Nakazono, and T.~Tsuda, \cite{KajNakTsu:2011:PRDPSTA}. 
For the projective reduction scenario, the element of the extended affine
Weyl group corresponding to the dynamics is only a \emph{quasi-translation} (i.e., it is an element of infinite order that becomes a 
translation after being raised to some power). However, it is possible to choose a sub-family in the generic family, by imposing some parameter constraints,
so that the dynamics becomes a translation, both in terms of the evolution of the coefficients in the equation, and in terms of the 
symmetry group of the sub-family. 
The existence of nodal curves also corresponds to parameter constraints, and we note that combinations of constraints from both projective reduction and nodal curves are possible. 
Here we adopt the point of view that the symmetry group of a \emph{discrete} Painlev\'e equation is that of the surface (sub-)family forming its configuration space, and whose 
translation elements generate the resulting dynamics. 

We illustrate this situation by considering one of the most well-known and well-studied examples of discrete Painlev\'e equations, a 
discrete d-$\Pain{II}$ equation,
\begin{equation}\label{eq:dP-II-eq}
	x_{n+1} + x_{n-1} = \frac{(\alpha n + \beta)x_{n} + \gamma}{1 - x_{n}^{2}},
\end{equation}
where $\alpha$, $\beta$, and $\gamma$ are some complex parameters.
According to \cite{GraRam:2004:DPER}, the $\gamma=0$ case of equation \eqref{eq:dP-II-eq} was first identified as a discrete analogue of $\Pain{II}$ in \cite{NijPap:1991:SRILDAPE}, after having appeared around the same time in \cite{periwalshevitz}, and its continuous limit to a special case of $\Pain{II}$ taken.

We show that the d-$\Pain{II}$ equation \eqref{eq:dP-II-eq} is a discrete dynamical system on a sub-family of the Sakai $D_{5}^{(1)}$ surface family whose actual symmetry group is the proper subgroup
$\widetilde{W}\big(A_{1}^{(1)}\big) \times \widetilde{W}\big(A_{1}^{(1)}\big)$ of the full group 
$\widehat{W}\big(A_{3}^{(1)}\big)$ of Cremona isometries of the full $D_{5}^{(1)}$ surface family --- the corresponding group element is only a quasi-translation in $\widehat{W}\big(A_{3}^{(1)}\big)$ but becomes a proper translation  in $\widetilde{W}\big(A_{1}^{(1)}\big) \times \widetilde{W}\big(A_{1}^{(1)}\big)$.

The paper is organized as follows. In the next section we give a brief summary of the algebro-geometric data for the standard realization 
of the $D_{5}^{(1)}$-family, following \cite{KajNouYam:2017:GAPE}. In Section~\ref{sec:dP-II} we explain how equation \eqref{eq:dP-II-eq}
fits into this framework, what is the parameter constraint that defines the sub-family, and what is its symmetry group. 
Here we would like to stress again that by the symmetry group of a discrete Painlev\'e equation we mean an extended affine Weyl group whose birational representation on a surface sub-family \emph{generates} the equation, rather than the discrete symmetries of the equation itself.
In Section~\ref{sec:dP-34} we consider an example of recurrence that appeared in a recent paper by Chao Ming and Liwei Wang, \cite{MinWan:2025:DPXHAFPDFRME}. That example is \eqref{eq:dP-II-eq} with an additional constraint on the parameters, which results in the appearance of a so-called \textit{nodal curve}, thus further restricting the symmetry group of the equation. This example is particularly interesting since it combines two different types of parameter constraints. In the final section we give a brief summary and formulate our conclusions.


\section{The Algebro-Geometric Data and Discrete Painlev\'e Equations on the $D_{5}^{(1)}$ Sakai Surface Family} 
\label{sec:D5-family}

\subsection{Geometric Realization} 
\label{sub:geometric_realization}
To make the paper self-contained, in this section we reproduce some standard facts about the $D_{5}^{(1)}$-Sakai surface family, following \cite{KajNouYam:2017:GAPE}, 
and review some standard examples of discrete Painlev\'e equations on that surface. This was also considered in detail in \cite[Appendix]{CheDzhHu:2020:PLUEDPE}, so here we only 
collect some essential information. Surfaces in this family 
are characterized by the condition that the configuration of the irreducible components $d_{i}$ 
(equivalently, their classes $\delta_{i}=[d_i]\in \operatorname{Pic}(X)$ called the \emph{surface roots})
of the unique effective anti-canonical divisor is described by an affine Dynkin diagram of type $D_{5}^{(1)}$:

\begin{figure}[ht]
\begin{equation}\label{eq:d-roots-d51}			
 	\raisebox{-32.1pt}{\begin{tikzpicture}[
 			elt/.style={circle,draw=black!100,thick, inner sep=0pt,minimum size=2mm}]
 		\path 	(-1,1) 	node 	(d0) [elt, label={[xshift=-10pt, yshift = -10 pt] $\delta_{0}$} ] {}
 		        (-1,-1) node 	(d1) [elt, label={[xshift=-10pt, yshift = -10 pt] $\delta_{1}$} ] {}
 		        ( 0,0) 	node  	(d2) [elt, label={[xshift=-10pt, yshift = -10 pt] $\delta_{2}$} ] {}
 		        ( 1,0) 	node  	(d3) [elt, label={[xshift=10pt, yshift = -10 pt] $\delta_{3}$} ] {}
 		        ( 2,1) 	node  	(d4) [elt, label={[xshift=10pt, yshift = -10 pt] $\delta_{4}$} ] {}
 		        ( 2,-1) node 	(d5) [elt, label={[xshift=10pt, yshift = -10 pt] $\delta_{5}$} ] {};
 		\draw [black,line width=1pt ] (d0) -- (d2) -- (d1)  (d2) -- (d3) (d4) -- (d3) -- (d5);
 	\end{tikzpicture}} \qquad
	- \mathcal{K}_{X} = \delta = \delta_{0} + \delta_{1} + 2 \delta_{2} + 2 \delta_{3} + \delta_{4} + \delta_{5}.
 \end{equation}
	\caption{Affine Dynkin diagram $D_{5}^{(1)}$}
	\label{fig:d5-diag}
\end{figure}

There are different geometric realizations of this family, corresponding to different choices of the surface roots. For example, in \cite{KajNouYam:2017:GAPE}
the surface root basis is 
\begin{equation}\label{eq:d5-surf-roots}
	\begin{aligned}
		\delta_{0} &= \mathcal{E}_{1} - \mathcal{E}_{2}, &\quad  \delta_{2} &= \mathcal{H}_{1} - \mathcal{E}_{1} - \mathcal{E}_{3}, &\quad 
			\delta_{4} &= \mathcal{E}_{5} - \mathcal{E}_{6},\\		
		\delta_{1} &= \mathcal{E}_{3} - \mathcal{E}_{4} &\quad  \delta_{3} &= \mathcal{H}_{1} - \mathcal{E}_{5} - \mathcal{E}_{7} 	&\quad 
			\delta_{5} &= \mathcal{E}_{7} - \mathcal{E}_{8},		
	\end{aligned}
\end{equation}
and in \cite{Sak:2001:RSAWARSGPE} it is taken as
\begin{equation}\label{eq:d5-surf-roots-alt}
	\begin{aligned}
		\delta_{0} &= \mathcal{H}_{1} - \mathcal{E}_{1} - \mathcal{E}_{2}, &\quad  \delta_{2} &= \mathcal{E}_{2} - \mathcal{E}_{3}, &\quad 
			\delta_{4} &= \mathcal{E}_{5} - \mathcal{E}_{6},\\		
		\delta_{1} &= \mathcal{E}_{3} - \mathcal{E}_{4} &\quad  \delta_{3} &= \mathcal{H}_{2} - \mathcal{E}_{2} - \mathcal{E}_{5} 	&\quad 
			\delta_{5} &= \mathcal{H}_{1} - \mathcal{E}_{7} - \mathcal{E}_{8},		
	\end{aligned}
\end{equation}
The resulting surface families are equivalent via an explicit birational change of variables, as carefully explained in \cite[Appendix]{CheDzhHu:2020:PLUEDPE}.
We choose the surface root basis \eqref{eq:d5-surf-roots}. The standard basis of the symmetry roots $\alpha_{i}\in \operatorname{Pic}(X)$, 
$\alpha_{i}\bullet \delta_{j} = 0$, for this configuration is shown on Figure~\ref{fig:a-roots-a3}.

\begin{figure}[ht]
\centering
\begin{equation}\label{eq:a-roots-a31}			
	\raisebox{-32.1pt}{\begin{tikzpicture}[
			elt/.style={circle,draw=black!100,thick, inner sep=0pt,minimum size=2mm}]
		\path 	(-1,1) 	node 	(a0) [elt, label={[xshift=-10pt, yshift = -10 pt] $\alpha_{0}$} ] {}
		        (-1,-1) node 	(a1) [elt, label={[xshift=-10pt, yshift = -10 pt] $\alpha_{1}$} ] {}
		        ( 1,-1) node  	(a2) [elt, label={[xshift=10pt, yshift = -10 pt] $\alpha_{2}$} ] {}
		        ( 1,1) 	node 	(a3) [elt, label={[xshift=10pt, yshift = -10 pt] $\alpha_{3}$} ] {};
		\draw [black,line width=1pt ] (a0) -- (a1) -- (a2) --  (a3) -- (a0); 
		\draw [purple, dashed, line width = 0.5pt] (0,-1.5) -- (0,1.2)	node [pos=0,below] {\small $\sigma_{1}$};
		\draw [purple, dashed, line width = 0.5pt] (-1.5,-1.5) -- (1.2,1.2) node [pos=0,below left] {\small $\sigma_{2}$};
		\draw [purple, dashed, line width = 0.5pt] (1.5,-1.5) -- (-1.2,1.2) node [pos=0,below right] {\small $\sigma_{3}$};
	\end{tikzpicture}} \qquad
			\begin{alignedat}{2}
			\alpha_{0} &= \mathcal{H}_{2} - \mathcal{E}_{1} - \mathcal{E}_{2}, &\qquad  \alpha_{2} &= \mathcal{H}_{2} - \mathcal{E}_{3} - \mathcal{E}_{4},\\
			\alpha_{1} &= \mathcal{H}_{1} - \mathcal{E}_{5} - \mathcal{E}_{6}, &\qquad  \alpha_{3} &= \mathcal{H}_{1} - \mathcal{E}_{7} - \mathcal{E}_{8}.
			\\[5pt]
			\delta & = \mathrlap{\alpha_{0} + \alpha_{1} +  \alpha_{2} + \alpha_{3}.} 
			\end{alignedat}
\end{equation}
	\caption{The symmetry root basis for the standard $A_{3}^{(1)}$ symmetry sub-lattice}
	\label{fig:a-roots-a3}	
\end{figure}

Then after making some normalization choices, we can take the corresponding point configuration on $\mathbb{P}^1 \times \mathbb{P}^1$ with affine coordinates $(q,p)$, as
shown on Figure~\ref{fig:std-d5-sur} (this will be explained in detail in section \ref{sub:the_surface_family_for_the_d_pain_ii_dynamics}). 
Using the 
\emph{Period Map} $\chi: \operatorname{Span}\{\alpha_{i}\}\to \mathbb{C}$  defined through the standard symplectic form 
$\omega = dp\wedge dq$, we can introduce, for each symmetry root $\alpha_{i}$, 
a canonical parameter $a_{i} = \chi(\alpha_{i})$, known as the \emph{root variable}. These root variables satisfy the usual 
normalization condition $a_{0} + a_{1} + a_{2} + a_{3} = 1$ and parameterize the blowup points on 
Figure~\ref{fig:std-d5-sur} 
as follows:
\begin{equation*}
	p_{1}(\infty,-t)\leftarrow p_{2}(0,-a_{0})\quad p_{3}(\infty,0)\leftarrow p_{4}(0,-a_{2})\quad p_{5}(0,\infty)\leftarrow p_{6}(a_{1},0)\quad 
	p_{7}(1,\infty) \leftarrow p_{8}(a_{3},0)
\end{equation*}
where $t$ is an additional parameter, the notation for which reflects connections to differential Painlev\'e equations. 
This is the same parameterization of the point configuration as in section 8.2.18 of \cite{KajNouYam:2017:GAPE}. 
Now, allowing the root variables and parameter $t$ to vary, one obtains a family of surfaces $\mathcal{X}\ni X_{\boldsymbol{a}}$, parameterized by the root variables and the extra parameter: $\boldsymbol{a}=(a_0,a_1,a_2,a_3;t)$. 
The geometric realization of the surface family also carries data of the enumeration of the blowups in terms of the parameters, and this gives a natural identification of all $\operatorname{Pic}(X_{\boldsymbol{a}})$ into a single lattice which we denote $\operatorname{Pic}(\mathcal{X})$.

\begin{figure}[ht]
	\begin{tikzpicture}[>=stealth,basept/.style={circle, draw=red!100, fill=red!100, thick, inner sep=0pt,minimum size=1.2mm}]
	\begin{scope}[xshift=0cm,yshift=0cm]
	\draw [black, line width = 1pt] (-0.4,0) -- (2.9,0)	node [pos=0,left] {\small $H_{2}$} node [pos=1,right] {\small $p=0$};
	\draw [black, line width = 1pt] (-0.4,2.5) -- (2.9,2.5) node [pos=0,left] {\small $H_{2}$} node [pos=1,right] {\small $p=\infty$};
	\draw [black, line width = 1pt] (0,-0.4) -- (0,2.9) node [pos=0,below] {\small $H_{1}$} node [pos=1,above] {\small $q=0$};
	\draw [black, line width = 1pt] (2.5,-0.4) -- (2.5,2.9) node [pos=0,below] {\small $H_{1}$} node [pos=1,above] {\small $q=\infty$};
	\node (p3) at (2.5,0) [basept,label={[xshift = -8pt, yshift=-15pt] \small $p_{3}$}] {};
	\node (p4) at (3,0.5) [basept,label={[yshift=0pt] \small $p_{4}$}] {};
	\node (p1) at (2.5,1) [basept,label={[xshift = -8pt, yshift=-15pt] \small $p_{1}$}] {};
	\node (p2) at (3,1.5) [basept,label={[yshift=0pt] \small $p_{2}$}] {};
	\node (p5) at (0,2.5) [basept,label={[xshift = 8pt, yshift=0pt] \small $p_{5}$}] {};
	\node (p6) at (-.5,2) [basept,label={[yshift=-15pt] \small $p_{6}$}] {};
	\node (p7) at (1.5,2.5) [basept,label={[xshift = 8pt, yshift=0pt] \small $p_{7}$}] {};
	\node (p8) at (1,2) [basept,label={[yshift=-15pt] \small $p_{8}$}] {};
	\draw [red, line width = 0.8pt, ->] (p2) -- (p1);
	\draw [red, line width = 0.8pt, ->] (p4) -- (p3);
	\draw [red, line width = 0.8pt, ->] (p6) -- (p5);
	\draw [red, line width = 0.8pt, ->] (p8) -- (p7);	
	\end{scope}
	\draw [->] (6.5,1)--(4.5,1) node[pos=0.5, below] {$\operatorname{Bl}_{p_{1}\cdots p_{8}}$};
	\begin{scope}[xshift=9cm,yshift=0cm]
	\draw [red, line width = 1pt] (-0.4,0) -- (3.5,0)	node [pos=0, left] {\small $H_{2}-E_{3}$};
	\draw [red, line width = 1pt] (0,-0.4) -- (0,2.4) node [pos=0, below] {\small $H_{1}-E_{5}$};
	\draw [blue, line width = 1pt] (-0.2,1.8) -- (0.8,2.8) node [pos=0, left] {\small $E_{5}-E_{6}$};
	\draw [red, line width = 1pt] (-0.1,2.7) -- (0.4,2.2) node [pos=0, above] {\small $E_{6}$};
	\draw [blue, line width = 1pt] (1.2,1.8) -- (2.2,2.8) node [pos=0, xshift=-14pt, yshift=-5pt] {\small $E_{7}-E_{8}$};
	\draw [red, line width = 1pt] (1.6,2.4) -- (2.1,1.9) node [pos=1, below] {\small $E_{8}$};
	\draw [blue, line width = 1pt] (0.3,2.6) -- (4.2,2.6) node [pos=1,right] {\small $H_{2} - E_{5} - E_{7}$};
	\draw [blue, line width = 1pt] (3,-0.2) -- (4,0.8) node [pos=1,right] {\small $E_{3} - E_{4}$};
	\draw [red, line width = 1pt] (3.4,0.4) -- (3.9,-0.1) node [pos=1, below] {\small $E_{4}$};
	\draw [blue, line width = 1pt] (3.8,0.3) -- (3.8,3) node [pos=1, above] {\small $H_{1}-E_{1} - E_{3}$};	
	\draw [blue, line width = 1pt] (3,1) -- (4,2) node [pos=1,right] {\small $E_{1} - E_{2}$};
		\draw [red, line width = 1pt] (3.1,1.9) -- (3.6,1.4) node [pos=0, above] {\small $E_{2}$};
	\draw [red, line width = 1pt] (-0.4,1.2) -- (3.5,1.2)	node [pos=0, left] {\small $H_{2}-E_{1}$};
	\draw [red, line width = 1pt] (1.4,-0.4) -- (1.4,2.4) node [pos=0, below] {\small $H_{1}-E_{6}$};
	\end{scope}
	\end{tikzpicture}
	\caption{A standard realization of the $D_{5}^{(1)}$ surface. Here we use the actual divisors rather than their divisor classes. 
    For example, $H_{1} - E_{1} - E_{3}$ is the proper transform of the line $q=\infty$ under the blowup procedure, and it is the unique 
    effective divisor in the class $\delta_{2} = \mathcal{H}_{1} - \mathcal{E}_{1} - \mathcal{E}_{3}$, and so on}
	\label{fig:std-d5-sur}
\end{figure}
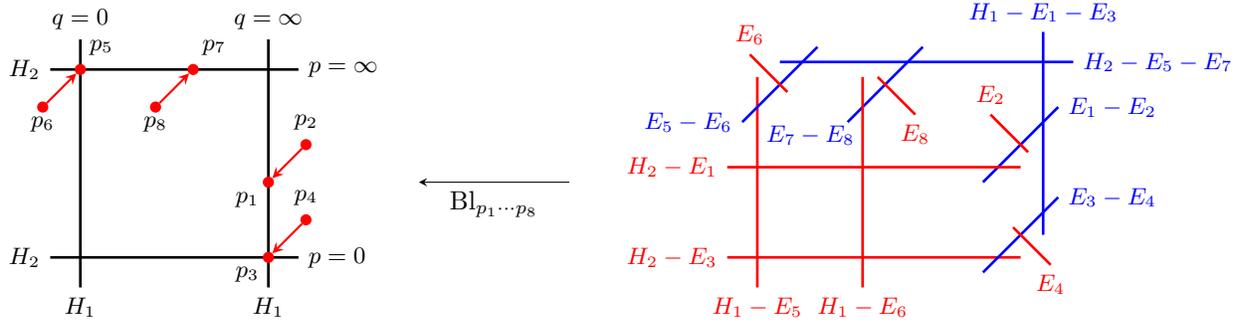


\subsection{Affine Weyl Symmetry Group} 
\label{sub:affine_weyl_symmetry_group}
We call elements $\alpha_{i}$ defined on 
Figure~\ref{fig:a-roots-a3} roots because they give a set of simple roots for the standard affine 
$A_{3}^{(1)}$ root system
in the space 
$V^{(1)}=\operatorname{Span}_{\mathbb{R}}\{\alpha_{0},\ldots,\alpha_{3}\}\subset \operatorname{Pic}^{\mathbb{R}}(\mathcal{X}):=\operatorname{Pic}(\mathcal{X})\otimes_{\mathbb{Z}} \mathbb{R}$ 
equipped with the symmetric bilinear form 
$(~\,,~)$ defined on the basis elements 
$\alpha_{i}$ in terms of the intersection product 
on $\operatorname{Pic}(\mathcal{X})$, 
$(v_1,v_2)=-v_1 \bullet v_2$. In particular, we get the standard affine Cartan matrix 
\begin{equation}\label{CA31}
   C\left(A_3^{(1)}\right)= \left(-\alpha_{i}\bullet\alpha_{j}\right) = 
   \left(
\begin{array}{cccc}
 2 & -1 & 0 & -1 \\
 -1 & 2 & -1 & 0 \\
 0 & -1 & 2 & -1 \\
 -1 & 0 & -1 & 2 \\
\end{array}
\right).
\end{equation}

The abstract affine Weyl group $W\left(A_{3}^{(1)}\right)$, defined in terms of generators 
$w_{i} = w_{\alpha_{i}}$ and relations that 
are encoded by the affine Dynkin diagram $A_{3}^{(1)}$,
\begin{equation}\label{DynA31}
	W\left(A_{3}^{(1)}\right) = W\left(\raisebox{-20pt}{\begin{tikzpicture}[
			elt/.style={circle,draw=black!100,thick, inner sep=0pt,minimum size=1.5mm}]
		\path 	(-0.5,0.5) 	node 	(a0) [elt, label={[xshift=-10pt, yshift = -10 pt] $\alpha_{0}$} ] {}
		        (-0.5,-0.5) node 	(a1) [elt, label={[xshift=-10pt, yshift = -10 pt] $\alpha_{1}$} ] {}
		        (0.5,-0.5) 	node 	(a2) [elt, label={[xshift=10pt, yshift = -10 pt] $\alpha_{2}$} ] {}
		        (0.5,0.5) 	node 	(a3) [elt, label={[xshift=10pt, yshift = -10 pt] $\alpha_{3}$} ] {};
		\draw [black,line width=1pt ] (a0) -- (a1) -- (a2)--(a3) -- (a0); 
	\end{tikzpicture}} \right)
	=
	\left\langle w_{0},\dots, w_{3}\ \left|\ 
	\begin{alignedat}{2}
    w_{i}^{2} = e,\quad  w_{i} w_{j} &= w_{j} w_{i}& &\text{ when 
   				\raisebox{-0.08in}{\begin{tikzpicture}[
   							elt/.style={circle,draw=black!100,thick, inner sep=0pt,minimum size=1.5mm}]
   						\path   ( 0,0) 	node  	(ai) [elt] {}
   						        ( 0.5,0) 	node  	(aj) [elt] {};
   						\draw [black] (ai)  (aj);
   							\node at ($(ai.south) + (0,-0.2)$) 	{$\alpha_{i}$};
   							\node at ($(aj.south) + (0,-0.2)$)  {$\alpha_{j}$};
   							\end{tikzpicture}}}\\
    w_{i} w_{j} w_{i} &= w_{j} w_{i} w_{j}& &\text{ when 
   				\raisebox{-0.17in}{\begin{tikzpicture}[
   							elt/.style={circle,draw=black!100,thick, inner sep=0pt,minimum size=1.5mm}]
   						\path   ( 0,0) 	node  	(ai) [elt] {}
   						        ( 0.5,0) 	node  	(aj) [elt] {};
   						\draw [black] (ai) -- (aj);
   							\node at ($(ai.south) + (0,-0.2)$) 	{$\alpha_{i}$};
   							\node at ($(aj.south) + (0,-0.2)$)  {$\alpha_{j}$};
   							\end{tikzpicture}}}
	\end{alignedat}\right.\right\rangle. 
\end{equation} 
can now be realized via \textit{reflections} in the 
roots $\alpha_{i}$, $w_{i} = r_{\alpha_{i}}$, which can be defined on the whole of $\operatorname{Pic}(\mathcal{X})$,
\begin{equation}\label{eq:root-refl}
	w_{i}(\mathcal{C}) = r_{\alpha_{i}}(\mathcal{C}) = \mathcal{C} - 2 
	\frac{\mathcal{C}\bullet \alpha_{i}}{\alpha_{i}\bullet \alpha_{i}}\alpha_{i}
	= \mathcal{C} + \left(\mathcal{C}\bullet \alpha_{i}\right) \alpha_{i},\qquad \mathcal{C}\in \operatorname{Pic}(\mathcal{X}).
\end{equation}

Then $\Delta^{(1)}=(\alpha_0,\alpha_1,\alpha_2,\alpha_3)$ is the
\emph{simple system} of the $A_3^{(1)}$ root system, reflections  $w_{i} = r_{\alpha_{i}}$ are called \emph{simple reflections}, and 
the affine $A_3^{(1)}$ root system is $\Phi^{(1)}=W^{(1)}\cdot\Delta^{(1)}$, where $W^{(1)}=W\left(A_3^{(1)}\right)$. We denote by $\Delta = (\alpha_1,\alpha_2,\alpha_3)$ the simple roots of the underlying $A_3$ root system
$\Phi=W\cdot\Delta$, where 
$W=W\left(A_3\right)=\langle w_{1},w_{2},w_{3}\rangle$. 

The anti-canonical divisor class $-\mathcal{K}_{\mathcal{X}}$ is in $V^{(1)}$ and, in fact, is the \emph{null root} $\delta$ of the $A_{3}^{(1)}$ root system,
\begin{equation}
-\mathcal{K}_{\mathcal{X}}=\delta=\alpha_0+\alpha_1+\alpha_2+\alpha_3=\alpha_0+\tilde{\alpha}    
\end{equation}
where $\tilde{\alpha}=\alpha_1+\alpha_2+\alpha_3$ is the highest root of the underlying $A_3$ root system.

Let $V = \operatorname{Span}_{\mathbb{R}}\{\alpha_{1},\alpha_{2},\alpha_{3}\}$. Recall that in the $A_{3}$ case, the fundamental weights $h_{i}\in V$ are 
given explicitly in terms of the simple roots by 
\begin{equation}\label{A3w}
    h_1 = \frac{3}{4} \alpha_1 + \frac{1}{2} \alpha_{2} + \frac{1}{4} \alpha_3, \quad
h_2 = \frac{1}{2} \alpha_1 +  \alpha_{2} + \frac{1}{2} \alpha_3, 
\quad
h_3 = \frac{1}{4} \alpha_1 + \frac{1}{2} \alpha_{2} + \frac{3}{4} \alpha_3.
\end{equation}
Then in $V$ we have two important lattices, the root lattice 
$Q=\operatorname{Span}_{\mathbb{Z}}\{\alpha_1,\alpha_2,\alpha_3\}$ and the weight lattice 
$P = \operatorname{Span}_{\mathbb{Z}}\{h_1,h_2,h_3\}$
of the finite $A_{3}$ system.

For any element 
$t\in \operatorname{Pic}^{\mathbb{R}}(\mathcal{X})$
we can define an \emph{associated translation} $\mathbf{T}_{t}$ 
by
\begin{equation}
    \mathbf{T}_{t}: \operatorname{Pic}^{\mathbb{R}}(\mathcal{X}) \to \operatorname{Pic}^{\mathbb{R}}(\mathcal{X}),\qquad
    \mathbf{T}_{t}(v) = v - (t\bullet v)\delta.
\end{equation}
These translations satisfy 
$\mathbf{T}_{t}\mathbf{T}_{t'} = \mathbf{T}_{t+t'}$
and, for any automorphism $w$ of $\operatorname{Pic}(\mathcal{X})$ preserving the intersection form and 
$\delta$, $\mathbf{T}_{w(t)}=w \mathbf{T}_{t} w^{-1}$.
When $t = \alpha\in Q$, this gives the usual translations 
$\mathbf{T}_{\alpha}(v) = v + (\alpha,v)\delta$
on the root lattice $Q$ and the standard fact is that 
\begin{equation}
    W\left(A_{3}^{(1)}\right) = W(A_{3}) \ltimes \mathbf{T}_{Q},\qquad \mathbf{T}_{Q} = \left\{\mathbf{T}_{\alpha}~|~\alpha \in Q\right\},
\end{equation}
and the semi-direct product structure is realized explicitly via
    $w_0 =  r_{\tilde{\alpha}} \mathbf{T}_{\tilde{\alpha}}$,
where the reflection corresponding to the highest root $\tilde{\alpha}$, written in terms of generators, is $r_{\tilde{\alpha}} = w_3 w_1 w_2 w_1 w_3$.

However if we take translations associated to $t=h\in P$, the same construction results in a larger group known as the \emph{extended} affine Weyl group. 
In the case of a finite crystallographic root system, it is known that the quotient $P/Q$ is a finite abelian group which corresponds to some but not necessarily all automorphisms of the affine Dynkin diagram, see \cite[VI]{bourbaki} for descriptions of these finite groups for the Dynkin diagrams of finite type root systems.

In our case, $\operatorname{Aut}\left(A_{3}^{(1)}\right)\simeq \mathbb{D}_{4}$, the dihedral group of order 8, which can be 
generated by reflections $\sigma_{1}$ and $\sigma_{2}$ shown on Figure~\ref{fig:a-roots-a3}, but it is convenient to include one more automorphism $\sigma_{3}$, see Figure~\ref{fig:a-roots-a3}. 
	These act on the symmetry and the surface root bases as permutations (here we use the standard cycle 
	notation):
	\begin{equation}
		\sigma_{1} = (\alpha_{0}\alpha_{3})(\alpha_{1}\alpha_{2})=
		(\delta_{0}\delta_{5})(\delta_{1}\delta_{4})(\delta_{2}\delta_{3})\qquad 
		\sigma_{2} = (\alpha_{0}\alpha_{2})=(\delta_{0}\delta_{1}) \qquad
		\sigma_{3} = (\alpha_{1}\alpha_{3})=(\delta_{4}\delta_{5}).
	\end{equation}
	They can also be represented as compositions of reflections (that are no longer in roots in the $A_3^{(1)}$ system but rather in the larger $E_8^{(1)}$ system containing it) when acting on the Picard lattice,
	\begin{equation}
			\sigma_{1} = (\mathcal{E}_{1}\mathcal{E}_{7})(\mathcal{E}_{2}\mathcal{E}_{8})(\mathcal{E}_{3}\mathcal{E}_{5})
			(\mathcal{E}_{4}\mathcal{E}_{6})w_{\mu},\qquad 
			\sigma_{2} = (\mathcal{E}_{1}\mathcal{E}_{3})(\mathcal{E}_{2}\mathcal{E}_{4}), \qquad
			\sigma_{3} = (\mathcal{E}_{5}\mathcal{E}_{7})(\mathcal{E}_{6}\mathcal{E}_{8}),			
	\end{equation}
	where $w_{\mu}$ is a reflection  \eqref{eq:root-refl} in $\mu = \mathcal{H}_{1} - \mathcal{H}_{2}$ 
	(note also that a transposition 
	$(\mathcal{E}_{i} \mathcal{E}_{j})$ is induced by a reflection in the root $\mathcal{E}_{i} - \mathcal{E}_{j}$).

The  automorphisms corresponding to $P/Q$
are only the rotations  $\Sigma= \langle \rho = \sigma_1 \sigma_2 \rangle \cong \mathbb{Z}_{4} 
\triangleleft\mathbb{D}_{4} \cong \operatorname{Aut}\left(A_3^{(1)}\right)$. 
Then we get the 
\emph{extended} affine Weyl group
\begin{equation} \label{extendedisoms1}
\widetilde{W}\left(A_3^{(1)}\right) = W\left(A_3^{(1)}\right)\rtimes \Sigma = W\left(A_3\right)\ltimes \mathbf{T}_{P},\qquad
\mathbf{T}_{P} = \left\{\mathbf{T}_{h}~|~h \in P\right\}.
\end{equation}
In the case at hand the equality in \eqref{extendedisoms1} is realised by 
\begin{equation}
        \rho   = \sigma_{1} \sigma_2 = w_{1} w_{2} w_{3} \mathbf{T}_{h_3},
        \quad 
        \rho^2 =  \sigma_{1} \sigma_2  \sigma_{1} \sigma_2 = w_{2} w_{3} w_{1} w_2 \mathbf{T}_{h_2} ,
        \quad
        \rho^3 =  \sigma_{2} \sigma_{1} 
        = w_{3} w_{2} w_{1} \mathbf{T}_{h_{1}},
\end{equation}
with the translations by fundamental weights $h_i$ acting by 
\begin{equation}
    \mathbf{T}_{h_i} (\alpha_0)= \alpha_0-\delta, \quad \mathbf{T}_{h_i}(\alpha_i) = \alpha_i +\delta, \quad \mathbf{T}_{h_i}(\alpha_j)=\alpha_j ~\text{ for }~ j \neq i.
\end{equation} 
This also gives us the expressions in terms of generators 
\begin{equation}
         \mathbf{T}_{h_1} = \rho^3 w_{2} w_{3} w_{0}, 
         \quad 
         \mathbf{T}_{h_2} = \rho^2 w_0 w_3 w_1 w_0, 
         \quad 
         \mathbf{T}_{h_3} = \rho w_2 w_1 w_0.
\end{equation}
At the same time,
for the geometric Painlev\'e theory 
we need to include all of the diagram automorphisms. 
To keep track of this distinction, we introduce the 
following terminology.
\begin{definition}\label{def:fully-extended-W}
The \emph{fully extended}
affine Weyl group (of type $A_{3}^{(1)}$)
is $\widehat{W}\left(A_{3}^{(1)}\right) := W\left(A_{3}^{(1)}\right) \rtimes \operatorname{Aut}\left(A_{3}^{(1)}\right)$.

\end{definition}

The semi-direct product structure of $\widehat{W}\left(A_3^{(1)}\right)$ is given by the action of $\sigma \in \operatorname{Aut}\left(A_{3}^{(1)}\right)$ on 
$W\left(A_{3}^{(1)}\right)$ via $w_{\sigma(\alpha_{i})} = \sigma w_{\alpha_{i}} \sigma^{-1}$.

The group $\widehat{W}\left(A_3^{(1)}\right)$ describes the symmetries of the surface family constructed in Section \ref{sec:D5-family}.
This is via an action of $\widehat{W}\left(A_3^{(1)}\right)$ on point configurations by elementary birational maps on $(q,p)$ and root variables $\boldsymbol{a}$ (which lift to isomorphisms $w : X_{\boldsymbol{a}} \rightarrow X_{w.\boldsymbol{a}}$, which can also be thought of as automorphisms $w: \mathcal{X}\to \mathcal{X}$ of the family of surfaces, which induce the linear actions of $w$ on $\operatorname{Pic}(\mathcal{X})$ by pullback or pushforward depending on convention).
This is known as a birational representation of $\widehat{W}\left(A_{3}^{(1)}\right)$, and the action of $\widehat{W}\left(A_3^{(1)}\right)$ by automorphisms of $\mathcal{X}$ is called the \emph{Cremona action} \cite{Sak:2001:RSAWARSGPE}. 
We describe this birational representation in the following Lemma \cite[Section A.3]{CheDzhHu:2020:PLUEDPE}.

\begin{lemma}\label{thm:bir-weyl-d5} The birational representation of $\widehat{W}\left(A_{3}^{(1)}\right)$, written in the affine 
	$(q,p)$-chart and the root variables $a_{i}$, is the following.  
	Reflections $w_{i}$ on $\operatorname{Pic}(\mathcal{X})$ are induced by the elementary 
	birational mappings, also denoted by $w_{i}$,
	\begin{alignat}{2}
		w_{0}&: 
		\left(\begin{matrix} a_{0} & a_{1} \\ a_{2} & a_{3} \end{matrix}\ ;\ t\ ;
		\begin{matrix} q \\ p \end{matrix}\right) 
		&&\mapsto 
		\left(\begin{matrix} -a_{0} & a_{0} + a_{1} \\ a_{2} & a_{0} + a_{3} \end{matrix}\ ;\ t\ ;
		\begin{matrix} \displaystyle q + \frac{a_{0}}{p + t} \\ p \end{matrix}\right), \\
		w_{1}&: \left(\begin{matrix} a_{0} & a_{1} \\ a_{2} & a_{3} \end{matrix}\ ;\ t\ ;
		\begin{matrix} q \\ p \end{matrix}\right)
		&&\mapsto 
		\left(\begin{matrix} a_{0} + a_{1} & -a_{1} \\ a_{1} + a_{2} & a_{3} \end{matrix}\ ;\ t\ ;
		\begin{matrix}  q \\ \displaystyle p - \frac{a_{1}}{q} \end{matrix}\right), \\
		w_{2}&: 
		\left(\begin{matrix} a_{0} & a_{1} \\ a_{2} & a_{3} \end{matrix}\ ;\ t\ ;
		\begin{matrix} q \\ p \end{matrix}\right) 
		&&\mapsto 
		\left(\begin{matrix} a_{0} & a_{1} + a_{2} \\ -a_{2} & a_{2} + a_{3} \end{matrix}\ ;\ t\ ;
		\begin{matrix} \displaystyle q + \frac{a_{2}}{p}\\ p \end{matrix}\right), \\
		w_{3}&: 
		\left(\begin{matrix} a_{0} & a_{1} \\ a_{2} & a_{3} \end{matrix}\ ;\ t\ ;
		\begin{matrix} q \\ p \end{matrix}\right) 
		&&\mapsto 
		\left(\begin{matrix} a_{0}+a_{3} & a_{1} \\ a_{2}+a_{3} & -a_{3} \end{matrix}\ ;\ t\ ;
		\begin{matrix} q \\ \displaystyle  p - \frac{a_{3}}{q-1} \end{matrix}\right).
	\end{alignat}	
    Note that the parameter $t$ can also change when we consider the Dynkin diagram automorphisms, so it is convenient to include it 
   	among the root variables.
	The actions of the generators $\sigma_1,\sigma_2$ of $\operatorname{Aut}\left(A_3^{(1)}\right)$, as well as $\sigma_3=\sigma_1\sigma_2\sigma_1$, are given by the following birational mappings:
	\begin{alignat}{2}
		\sigma_{1}&: 
		\left(\begin{matrix} a_{0} & a_{1} \\ a_{2} & a_{3} \end{matrix}\ ;\ t\ ;
		\begin{matrix} q \\ p \end{matrix}\right) 
		&&\mapsto 
		\left(\begin{matrix} a_{3} & a_{2} \\ a_{1} & a_{0}  \end{matrix}\ ;\ -t\ ;
		\begin{matrix} \displaystyle -\frac{p}{ t} \\ q t \end{matrix}\right), \\
		\sigma_{2}&: 
		\left(\begin{matrix} a_{0} & a_{1} \\ a_{2} & a_{3} \end{matrix}\ ;\ t\ ;
		\begin{matrix} q \\ p \end{matrix}\right) 
		&&\mapsto 
		\left(\begin{matrix} a_{2} & a_{1} \\ a_{0} & a_{3}  \end{matrix}\ ;\ -t\ ;
		\begin{matrix} q \\ p + t \end{matrix}\right), \\
		\sigma_{3}&: 
		\left(\begin{matrix} a_{0} & a_{1} \\ a_{2} & a_{3} \end{matrix}\ ;\ t\ ;
		\begin{matrix} q \\ p \end{matrix}\right) 
		&&\mapsto 
		\left(\begin{matrix} a_{0} & a_{3} \\ a_{2} & a_{1}  \end{matrix}\ ;\ -t\ ;
		\begin{matrix} 1-q \\ -p  \end{matrix}\right).
	\end{alignat}

\end{lemma}


\subsection{Examples of Discrete Painlev\'e Equations} 
\label{sub:std-examples-d5}
There are two standard examples of discrete Painlev\'e equations on this surface family. 
The first one is \cite[(8.23)]{KajNouYam:2017:GAPE}
and it corresponds to the translation $\mathbf{T}_{h_{1}-h_{2}+h_{3}}$, i.e., 
its action on the symmetry roots by pushforward is given by 
\begin{equation}\label{eq:dPD5-transl-KNY}
	\psi_{*}: (\alpha_{0}, \alpha_{1}, \alpha_{2}, \alpha_{3})
	\mapsto (\alpha_{0}, \alpha_{1}, \alpha_{2}, \alpha_{3})
    + ( -1,1,-1,1 ) \delta,
\end{equation}
This equation can be written as $\psi : (q,p)\mapsto(\bar{q},\bar{p})$, where
\begin{equation}\label{eq:dPD5-KNY}
	\overline{q} + q = 1 - \frac{a_{2}}{p} - \frac{a_{0}}{p+t},\qquad 
	p + \underline{p} = -t + \frac{a_{1}}{q} + \frac{a_{3}}{q-1},
\end{equation}
which gives an isomorphism $\psi: X_{\boldsymbol{a}}\rightarrow X_{\bar{\boldsymbol{a}}}$, where the root variable evolution and normalization are given by 
\begin{equation}\label{eq:dPD5-rv-evol}
	\overline{a}_{0} = a_{0} + 1, \quad \overline{a}_{1} = a_{1}-1, \quad \overline{a}_{2} = a_{2} + 1,\quad \overline{a}_{3} = a_{3} - 1,\qquad
	 a_{0} + a_{1} + a_{2}  + a_{3} = 1.
\end{equation}
Note comparing \eqref{eq:dPD5-transl-KNY} and \eqref{eq:dPD5-rv-evol} that there is a correspondence between actions on root variables and simple roots, but the action by pushforward of the map on simple roots is inverse to the evolution of root variables, which is explained in terms of the definition of the period map in, for example, \cite{CarDzhTak:2017:FDIMDPE}. 
As is often the case, equations \eqref{eq:dPD5-KNY} naturally define two \emph{half-maps}, $\psi_{1}: (q,p)\to(\overline{q},-p)$ and 
$\psi_{2}:(q,p)\to(q,-\underline{p})$ (the additional negative sign here is related to the M\"obius group gauge action) and the full
mapping $\psi:(q,p)\mapsto (\overline{q},\overline{p})$ decomposes as $\psi = (\overline{\psi_{2}})^{-1}\circ \psi_{1}$, where $\overline{\psi}_2$ is $\psi_2$ above with root variables evolved according to \eqref{eq:dPD5-rv-evol}.

In terms of the action of generators of $\widehat{W}\left(A_{3}^{(1)}\right)$ on $(q,p)$ and root variables as in Lemma \ref{thm:bir-weyl-d5}, these mappings can be decomposed as 
\begin{equation}\label{eq:dPD5-decomp-KNY}
	\psi = \sigma_{3}\sigma_{2} w_{3} w_{1} w_{2} w_{0},\qquad \psi_{1} = \sigma_{3}w_{2}w_{0},\quad \psi_{2} = \sigma_{2}w_{3}w_{1}.
\end{equation}

The second example, called a d-$\Pain{IV}$ equation in \cite{Sak:2001:RSAWARSGPE}, is the mapping 
$\eta:(f,g)\to(\overline{f},\overline{g})$ that corresponds via pushforward to the translation $\mathbf{T}_{h_{3}}$. 
It is written in the \emph{multiplicative-additive} form
\begin{equation}\label{eq:dPD5-Sakai}
	\overline{f} f = \frac{s \overline{g}}{(\overline{g} - a_{3} + \lambda)(\overline{g} + a_{0} + \lambda)},\qquad 
	\overline{g} + g = \frac{s}{f} + \frac{a_{1}+ a_{0}}{1 - f} - \lambda + a_{3} - a_{0},
\end{equation}
where $\lambda = a_{0} + a_{1} + a_{2} + a_{3}$ (without loss of generality it can be normalized to $\lambda = 1$) 
the root variable evolution is given by $\overline{a}_{0} = a_{0} + \lambda$ and $\overline{a}_{3} = a_{3} - \lambda$, the action on 
the symmetry roots is 
\begin{equation}\label{eq:dPD5-transl-Sakai}
	\eta_{*}: ( \alpha_{0}, \alpha_{1}, \alpha_{2}, \alpha_{3} )
	\mapsto ( \alpha_{0}, \alpha_{1}, \alpha_{2}, \alpha_{3} ) + 
    ( -1,0,0,1)\delta,
\end{equation}
and the decomposition of the mapping in terms of the generators is 
$\eta = \sigma_{3}\sigma_{1}w_{2}w_{1}w_{0}$. 

Using Equations \eqref{CA31} and \eqref{A3w}, one
sees that $|h_{1}-h_{2}+h_{3}|^2=|h_{2}|^2=1$, while $|h_{3}|^2=3/4$, and therefore the associated translations $T_{h_{1}-h_{2}+h_{3}}$ and $T_{h_{3}}$ are not conjugate in $\widehat{W}\left(A_3^{(1)}\right)$ (since $\mathbf{T}_{w(h)}=w \mathbf{T}_{h} w^{-1}$ for any $h\in P$ and $w\in \widehat{W}\left(A_3^{(1)}\right)$ and $w$ preserves the intersection form). Hence $\psi$ and $\eta$ are not related under conjugation by any element of  $\widehat{W}\left(A_{3}^{(1)}\right)$.
So, the corresponding equations are not equivalent under change of variables corresponding to conjugation by any of the birational mappings in Lemma \ref{thm:bir-weyl-d5}.
Furthermore, 
equations \eqref{eq:dPD5-Sakai} correspond to 
a different geometric realization \eqref{eq:d5-surf-roots-alt},  but since our example is related to \eqref{eq:dPD5-transl-KNY},
we do not go into details here, see \cite[Appendix]{CheDzhHu:2020:PLUEDPE}.



\section{d-$\Pain{II}$ Equation} 
\label{sec:dP-II}

Let us now consider the d-$\Pain{II}$ equation \eqref{eq:dP-II-eq}.
First we need to show that the dynamics is indeed regularized on a family of $D_{5}^{(1)}$ surfaces. 
This is a standard computation that we only outline here, for a detailed description 
see, e.g., \cite{DzhFilSto:2020:RCDOPWHWDPE}.

\subsection{The Surface Family for the d-$\Pain{II}$ Dynamics} 
\label{sub:the_surface_family_for_the_d_pain_ii_dynamics}
We first rewrite \eqref{eq:dP-II-eq}  as a system,
\begin{equation}\label{eq:dP-II-sys}
	\left\{
	\begin{aligned}
	y_{n} &= x_{n+1}\\
	y_{n} + x_{n-1} &= \frac{(\alpha n + \beta)x_{n} + \gamma}{1 - x_{n}^{2}}		
	\end{aligned}
	\right. ,
\end{equation}
and then as a mapping (the parameter $\alpha$ should not be confused with a symmetry root, but the notation is traditional and the 
context makes it clear). Using the standard notation $\overline{x}:=x_{n+1}$, $\underline{x} := x_{n-1}$ and the same for $y$, 
the \emph{forward mapping} is given by 
\begin{equation}\label{eq:dP-II-fs}
	\varphi:(x,y)\mapsto (\overline{x},\overline{y}) = \left(y, \frac{(\alpha (n+1) + \beta) y + \gamma}{1 - y^{2}} - x \right)
\end{equation}
and the \emph{backward mapping} is
\begin{equation}\label{eq:dP-II-bs}
	\varphi^{-1}:(x,y)\mapsto (\underline{x},\underline{y}) = 
		\left( \frac{(\alpha n + \beta)x + \gamma}{1 - x^{2}} - y ,x\right).
\end{equation} 
Note that these mappings indeed define an additive non-autonomous discrete dynamics, since the coefficients in the mapping are 
(affine) functions of the time step $n$.

Next, extend the dynamics from $\mathbb{C}^{2}$ to $\mathbb{P}^{1} \times \mathbb{P}^{1}$ by introducing, in addition to the affine chart 
$(x,y)$, three more charts $(X,y)$, $(x,Y)$, and $(X,Y)$, where $X = 1/x$ and $Y = 1/y$. 
We then see the appearance of \emph{base points} where both the numerator and the denominator of the rational mapping vanish
simultaneously. For example, the forward mapping \eqref{eq:dP-II-fs}, when written in the $(X,y)$-chart in the domain, becomes
\begin{equation*}
	\varphi(X,y) = (\bar{x},\bar{y}) = \left(y, \frac{(\alpha (n+1) + \beta) X y  + \gamma X + y^{2}-1}{X(1 - y^{2})}\right)
\end{equation*}
and we immediately see the base points $(0,\pm1)$ in that chart. 

These indeterminacies of the mapping are resolved using the blowup procedure which, on the coordinate level, is just a change 
of variables. E.g., blowing up a point $q(x_{0},y_{0})$ introduces two charts $(u,v)$ and $(U,V)$ near this point via
\begin{equation*}
	x = x_{0} + u = x_{0} + UV,\qquad y = y_{0} + uv = y_{0} + V,
\end{equation*} 
where the coordinates $v = (y-y_{0})/(x-x_0)$ and $U = (x-x_{0})/(y-y_{0})$ are the \emph{slope} 
coordinates on the $\mathbb{P}^{1}$-``line 
of slopes'' or the exceptional divisor $E$ that we ``bubble'' at the point of the blowup. We then extend this algebraic 
mapping to the new chart, see if there are more base points that we need to blow up and continue this process 
until the mapping becomes an isomorphism after a finite number of blowups, which is always the case 
for the discrete Painlev\'e dynamics. 
In our example, the full set of the base points consists of the following four \emph{cascades} of infinitely near points:
\begin{equation}\label{eq:dP-II-points}
	\begin{aligned}
	q_{1}(x=\infty,y=-1)&\leftarrow q_{2}\left(u_{1}=0,v_{1}=\frac{\gamma-\alpha(n+1) - \beta}{2}\right)	\\
	q_{3}(x=\infty,y=1)&\leftarrow q_{4}\left(u_{3}=0,v_{3}=\frac{-\gamma-\alpha(n+1) - \beta}{2}\right)	\\
	q_{5}(x=-1,y=\infty)&\leftarrow q_{6}\left(U_{5}=\frac{\gamma-\alpha n - \beta}{2},V_{5}=0\right)	\\
	q_{7}(x=1,y=\infty)&\leftarrow q_{8}\left(U_{7}=\frac{-\gamma-\alpha n - \beta}{2},V_{7}=0\right).		
	\end{aligned}
\end{equation}

We see that, up to linear change of variables and parameter matching,
\begin{equation}\label{eq:var-chande-dP2}
	x = 2q - 1,\qquad y = \frac{\alpha}{2}p + 1,\qquad 
    \alpha = \frac{4}{t},
\end{equation}
the point configuration is exactly the one shown on Figure~\ref{fig:std-d5-sur}.
This also gives us the root variables in terms of the parameters $\alpha$, $\beta$, $\gamma$, and $n$:
\begin{equation}\label{eq:roots-vars-dp2}
	\begin{aligned}
	a_{0}&=\frac{n+1}{2} + \frac{\beta-\gamma}{2 \alpha}, &\qquad a_{2} &= \frac{n+1}{2} + \frac{\beta+\gamma}{2 \alpha} \\
	a_{1}&=-\frac{n}{2} - \frac{\beta-\gamma}{2 \alpha}, &\qquad a_{3} &= -\frac{n}{2} - \frac{\beta+\gamma}{2 \alpha}.
	\end{aligned}
\end{equation}
Note that in this case there are some relations among the root variables $a_i$.


\subsection{Dynamics on the Picard Lattice} 
\label{sub:dynamics_on_the_picard_lattice}
Let us now understand how the d-$\Pain{II}$ dynamics relates to the examples of standard discrete Painlev\'e equations considered in Section~\ref{sub:std-examples-d5}.
With the change of variables \eqref{eq:var-chande-dP2} and identification of parameters \eqref{eq:roots-vars-dp2}, we can consider the d-$\Pain{II}$ dynamics on (possibly, a proper sub-family of) the configuration space $\mathcal{X}$ from the geometric realization above.
Direct computation shows that the mapping $\varphi$ induces the linear map $\varphi_{*}:\operatorname{Pic}(\mathcal{X})\to \operatorname{Pic}(\mathcal{X})$ given by
\begin{alignat*}{2}
	\mathcal{H}_{1}& \mapsto 2 {\mathcal{H}}_{1} +  {\mathcal{H}}_{2} - {\mathcal{E}}_{5} - 
	 {\mathcal{E}}_{6} - {\mathcal{E}}_{7} - {\mathcal{E}}_{8}, \qquad&
	\mathcal{H}_{2}& \mapsto {\mathcal{H}}_{1},  \\
	\mathcal{E}_{1} &\mapsto {\mathcal{H}}_{1}  -{\mathcal{E}}_{6}, \qquad&
	\mathcal{E}_{5} &\mapsto {\mathcal{E}}_{3}\\
	\mathcal{E}_{2} &\mapsto {\mathcal{H}}_{1} - {\mathcal{E}}_{5},\qquad&
	\mathcal{E}_{6} &\mapsto {\mathcal{E}}_{4}\\
	\mathcal{E}_{3} &\mapsto {\mathcal{H}}_{1}  - {\mathcal{E}}_{8},\qquad&
	\mathcal{E}_{7} &\mapsto {\mathcal{E}}_{1}, \\
	\mathcal{E}_{4} &\mapsto {\mathcal{H}}_{1}  - {\mathcal{E}}_{7}, \qquad&
	\mathcal{E}_{8} &\mapsto {\mathcal{E}}_{2}.
\end{alignat*}
Hence we get the following action on the symmetry roots and the root variables (which is induced by the pull-back and so is inverse) 
\begin{align}
\varphi_{*}(  \alpha_{0}, \alpha_{1}, \alpha_{2}, \alpha_{3}  ) &=  
			( -\alpha_{1}, \alpha_{1} + \alpha_{2} + \alpha_{3} = \delta - \alpha_{0}, -\alpha_{3}, 
			\alpha_{0} + \alpha_{1} + \alpha_{3} = \delta - \alpha_{2}) \label{eq:dP-a-root-evol}	\\
			\overline{a}_{0} &= 1 - a_{1},\qquad \overline{a}_{1} = -a_{0},\qquad \overline{a}_{2} = 1 - a_{3}, \qquad \overline{a}_{3} = - a_{2}.	\label{eq:dP-a-pars-evol}		
\end{align}
This action therefore is \emph{not a \textbf{translation}} element of $\widehat{W}\left(A_{3}^{(1)}\right)$. It is, however, a \emph{quasi-translation}. Indeed,
it is a half of the standard translation \eqref{eq:dPD5-transl-KNY}:  $\varphi_{*}^{2} = \psi_{*}$. This can be seen either directly via the action on the symmetry roots, 
or by decomposing $\varphi_{*}$ in terms of the generators of the symmetry group,

\begin{equation}\label{eq:abs-dP-II}
	\varphi_{*} = \sigma_{1}\sigma_{2}w_{2}w_{0},
\end{equation}
and computing using the relations between generators of 
$\widehat{W}\left(A_{3}^{(1)}\right)$,
\begin{equation}\label{eq:abs-dP-II2}
	\varphi_{*}^{2} = \sigma_{1}\sigma_{2}w_{2}w_{0}\sigma_{1}\sigma_{2}w_{2}w_{0} = \sigma_{1}\sigma_{2}\sigma_{1}w_{1}w_{3}\sigma_{2}w_{2}w_{0}
	= \sigma_{3} \sigma_{2} w_{1} w_{3} w_{2} w_{0} = \sigma_{3} \sigma_{2} w_{3} w_{1} w_{2} w_{0} = \psi_{*}.
\end{equation}
\begin{definition}
	We call the equivalence class of \eqref{eq:abs-dP-II} in $\widehat{W}\left(A_{3}^{(1)}\right)$ an \emph{abstract d-$\Pain{II}$ equation}.
\end{definition}

Thus, the mapping $\varphi$ does not correspond to a non-autonomous additive difference equation, in the sense that the coefficients in the mapping cannot be written as affine functions of $n$. Indeed, the resulting equation, written for generic 
parameters $a_{0},\ldots,a_{3}$ constrained by $a_{0} + \cdots + a_{3} = 1$ becomes
\begin{equation}\label{eq:dp2-map-gen-pars}
	\varphi: 	 
		\left(\begin{matrix} a_{0} & a_{1} \\ a_{2} & a_{3} \end{matrix}\ ;\ \alpha\ ;
		\begin{matrix} x \\ y \end{matrix}\right) 
		\mapsto 
		\left(\begin{matrix} 1-a_{1} & -a_{0} \\ 1-a_{3} & -a_{2}  \end{matrix}\ ;\ \alpha\ ;
		\begin{matrix} \displaystyle y \\[3pt] -x - \displaystyle{\frac{\alpha a_{0}}{y+1}} - \frac{\alpha a_{2}}{y-1}\end{matrix}\right) .
\end{equation} 
If we regard the evolution of root variables under $\varphi$ as in \eqref{eq:dp2-map-gen-pars} as a system of difference equations for $a_0(n),a_1(n),a_2(n),a_3(n)$, its general solution is given by 
\begin{equation} \label{ndeprootvarsgeneric} 
    \begin{aligned}
        a_0(n) &=  \frac{n}{2} + \frac{(-1)^n-1}{4}+\frac{(-1)^n+1}{2} C_1 +\frac{(-1)^n-1}{2} C_2 
        = \begin{cases}
            \frac{n}{2} + C_1 & \text{if } n \text{ even}\\
            \frac{n-1}{2} - C_2 & \text{if } n \text{ odd}
        \end{cases},\\
        a_1(n) &= - \frac{n}{2} + \frac{(-1)^n+3}{4} +\frac{  (-1)^n-1}{2} C_1 +\frac{(-1)^n+1}{2} C_2 
        = \begin{cases}
            -\frac{n}{2} +1 + C_2 & \text{if } n \text{ even}\\
            -\frac{n-1}{2} - C_1 & \text{if } n \text{ odd}
        \end{cases},\\
        a_2(n) &= \frac{n}{2} +\frac{(-1)^n-1}{4} +\frac{(-1)^n+1}{2} C_3 +\frac{(-1)^n-1}{2} C_4 
        = \begin{cases}
            \frac{n}{2}  + C_3 & \text{if } n \text{ even}\\
            \frac{n-1}{2} - C_4 & \text{if } n \text{ odd}
        \end{cases}, \\
        a_3(n) &= - \frac{n}{2} + \frac{(-1)^n+3}{4} +\frac{ (-1)^n-1}{2} C_3 +\frac{(-1)^n+1}{2} C_4 
        = \begin{cases}
            -\frac{n}{2} +1 + C_4 & \text{if } n \text{ even}\\
            -\frac{n-1}{2} - C_3 & \text{if } n \text{ odd}
        \end{cases},
    \end{aligned}
\end{equation}
where $C_1,\dots,C_4$ are constants, subject to $C_1+C_2+C_3+C_4+1=0$ if we assume the normalization $a_{0} + \cdots + a_{3} = 1$.   
The dynamics defined by $\varphi$ in \eqref{eq:dp2-map-gen-pars} can then be written as the equation  
\begin{equation}\label{eq:dp2-eqn-gen-pars}
		x_{n+1} + x_{n-1} = \frac{\alpha ((a_{1}(n) + a_{3}(n))x - (a_{1}(n) - a_{3}(n)))}{x_n^{2}-1},	
\end{equation} 
with $a_i(n)$ given by \eqref{ndeprootvarsgeneric} and we see that the coefficients of the equation are no longer affine functions of $n$. 
The fact that, using a quasi-translation on the full surface family, the root variables allow one to still write down an equation with coefficients being explicit functions of $n$ was the point of the paper \cite{fullparameter}.

However, looking at the expressions of the root variables \eqref{eq:roots-vars-dp2} from the actual d-$\Pain{II}$ equation we observe that, independent of the parameter values, 
the root variables satisfy the constraint
\begin{equation}\label{eq:dp2-par-constr}
		a_{0} + a_{1} = a_{2} + a_{3} = \frac{1}{2}. 
\end{equation}
Restricting to the sub-locus of the surfaces in the whole family whose parameters satisfy this condition recovers the translational nature
of the dynamics. Indeed,
\begin{equation}\label{eq:par-consrt-evolve}
	\overline{a}_{0} = 1 - a_{1} = a_{0} + \frac{1}{2},\quad \overline{a}_{1} =  - a_{0} = a_{1} - \frac{1}{2},\quad 
	\overline{a}_{2} = 1 - a_{3} = a_{2} + \frac{1}{2},\quad \overline{a}_{3} = - a_{2} = a_{3} - \frac{1}{2},
\end{equation}
and so $\overline{a_{1} + a_{3}} = a_{1} + a_{3} - 1$, $\overline{a_{1} - a_{3}} = a_{1} - a_{3}$. Thus, $a_{1} + a_{3}$ is a linear and $a_{1} - a_{3}$ is a constant
function of $n$, exactly as in \eqref{eq:dP-II-eq}.

Since the evolution of the root variables in \eqref{eq:dp2-map-gen-pars} decouples into $(a_{0},a_{1})$ and $(a_{2},a_{3})$ pairs, we can 
visualize what happens by looking at them individually. On Figure~\ref{fig:par-dyn} we show the evolution of $(a_{0},a_{1})$ for
generic pair of parameters (left) and for a pair satisfying the constraint $a_{0} + a_{1} = \frac{1}{2}$ (right), cf \cite{KajNakTsu:2011:PRDPSTA}.

\begin{figure}[h]
\begin{tikzpicture}[basept/.style={circle, draw=red!100, fill=red!100, thick, inner sep=0pt, minimum size=1.2mm}]
\begin{axis}[width=3in, height=3in,
  xlabel={$a_{0}$},	
  ylabel={$a_{1}$},	
    ylabel near ticks,
  xmin=0, xmax = 6,
  ymin=-5, ymax=1,  
  axis equal,
  xtick={0,1,...,6},
  ytick={-5,-4,...,1},
  grid=both,
  grid style={line width=.4pt, draw=gray, dotted},
  major grid style={gray,line width=0.6, solid},
  minor tick num=4,
  title={ $a_{0} +a_{1} \neq \frac{1}{2}$},
]
 \addplot+ [sharp plot, thick, color=teal, mark = *, mark options={fill=teal}
    ] coordinates {
       (-0.6,0.8) (0.2,0.6) (0.4,-0.2) (1.2,-0.4) (1.4, -1.2) (2.2, -1.4) (2.4, -2.2) (3.2, -2.4) (3.4, -3.2) (4.2, -3.4) (4.4, -4.2) (5.2, -4.4) (5.4,-5.2)
    };
  \addplot [thick,  red, decoration={markings, mark=between positions 0.155 and 1 step 0.167 with {\arrow[red]{stealth}}},
    postaction=decorate]
          coordinates { (-0.6,0.8) (5.4, -5.2)};		
  \addplot [thick, orange, decoration={markings, mark=between positions 0.132 and 1 step 0.143 with {\arrow[orange]{stealth}}},
    postaction=decorate]
          coordinates { (-0.8,1.6) (6.2, -5.4)};		
\end{axis}
\end{tikzpicture} \hfill
\begin{tikzpicture}[basept/.style={circle, draw=red!100, fill=red!100, thick, inner sep=0pt, minimum size=1.2mm}]
\begin{axis}[width=3in, height=3in,
  xlabel={$a_{0}$},	
  ylabel={$a_{1}$},	
  ylabel near ticks,
  xmin=0, xmax = 6,
  ymin=-5, ymax=1,  
  axis equal,
  xtick={0,1,...,6},
  ytick={-5,-4,...,1},
  grid=both,
  grid style={line width=.4pt, draw=gray, dotted},
  major grid style={gray,line width=0.6, solid},
  minor tick num=4,
  title={ $a_{0} +a_{1} = \frac{1}{2}$},
]
 \addplot+ [sharp plot, thick, color=teal, mark = *, mark options={fill=teal}
    ] coordinates {
	 (-0.6, 1.1) (-0.1, 0.6) (0.4, 0.1) (0.9, -0.4) (1.4, -0.9) (1.9,  -1.4) (2.4, -1.9) (2.9, -2.4) (3.4, -2.9) (3.9, -3.4) (4.4,-3.9) 
	 (4.9, -4.4) (5.4, -4.9) (5.9, -5.4)
    };
  \addplot [thick, teal, decoration={markings, mark=between positions 0.068 and 1 step 0.077 with {\arrow[teal]{stealth}}},
    postaction=decorate]
          coordinates { (-0.6, 1.1) (5.9, -5.4)};		
\end{axis}
\end{tikzpicture}	
\caption{Parameter dynamics in the $(a_0,a_1)$ plane}
\label{fig:par-dyn}
\end{figure}
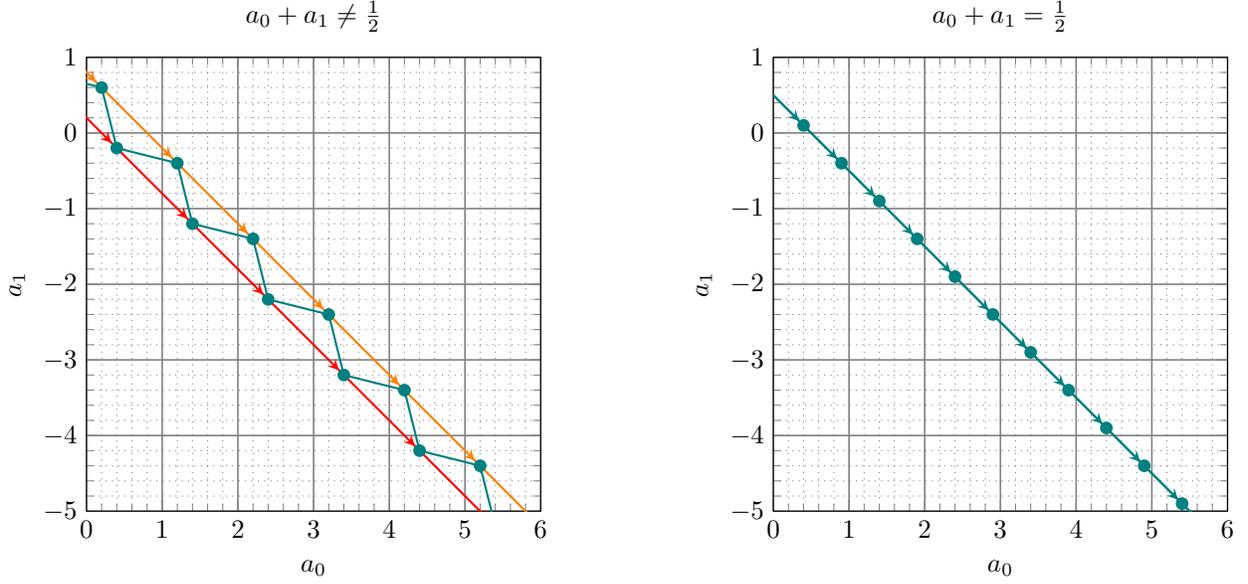
Thus, we are interested in the following question:
\begin{quote}
	What is the symmetry group of the sub-family of the full $D_{5}^{(1)}$-surface family that corresponds to the parameter 
	constraint \eqref{eq:dp2-par-constr}? Is it again an extended affine Weyl group? Does it generate the dynamics \eqref{eq:dP-II-fs}?
\end{quote}

We consider this question in the next sub-section. 


\subsection{The Symmetry Group of the Constrained Family} 
\label{sub:the_symmetry_group_of_the}

In this section we determine the subgroup of $\widehat{W}\left(A_{3}^{(1)}\right)$ such that its action 
on the root variables fixes the constraint \eqref{eq:dp2-par-constr}. That is, 
\begin{equation*}
a_{0} + a_{1} = a_{2} + a_{3} = \frac{1}{2} \quad \implies\quad	\overline{a}_{0} + \overline{a}_{1} = \overline{a}_{2} + \overline{a}_{3} = \frac{1}{2}.
\end{equation*} 
There are two possibilities:
\begin{equation*}
	\left\{ \begin{aligned}
		\overline{a}_{0} + \overline{a}_{1} &= a_{0} + a_{1} \\
		\overline{a}_{2} + \overline{a}_{3} &= a_{2} + a_{3} 
	\end{aligned}\right.
	\qquad \text{or} \qquad
\left\{ \begin{aligned}
	\overline{a}_{0} + \overline{a}_{1} &= a_{2} + a_{3} \\
	\overline{a}_{2} + \overline{a}_{3} &= a_{0} + a_{1} 
\end{aligned}\right.. 
\end{equation*}
Since $a_{i} = \chi(\alpha_{i})$, this parameter condition, on the level of the symmetry roots, becomes
\begin{equation*}
	\left\{ \begin{aligned}
		w(\alpha_{0} + \alpha_{1}) &= \alpha_{0} + \alpha_{1} \\
		w(\alpha_{2} + \alpha_{3}) &= \alpha_{2} + \alpha_{3}
	\end{aligned}\right.
	\  \text{or} \ 
\left\{ \begin{aligned}
		w(\alpha_{0} + \alpha_{1}) &= \alpha_{2} + \alpha_{3} \\
		w(\alpha_{2} + \alpha_{3}) &= \alpha_{0} + \alpha_{1}
\end{aligned}\right. \ \Leftrightarrow \ w\in \operatorname{Stab}\{\alpha_{0}+\alpha_{1},\alpha_{2} + \alpha_{3}\} < \widehat{W}\left(A_{3}^{(1)}\right).
\end{equation*}
We have the following main result describing the symmetry group of the 
d-$\Pain{II}$ equation.

\begin{theorem} \label{th:maintheorem}
The subgroup of $\widehat{W}(A_3^{(1)})$ compatible with the constraint on root variables coming from 
{\textup d-$\Pain{II}$} is 
	\begin{equation}
		G = \operatorname{Stab}\{\alpha_{0}+\alpha_{1},\alpha_{2} + \alpha_{3}\} \simeq 
		\widetilde{W}\left(A_{1}^{(1)}\right) \times \widetilde{W}\left(A_{1}^{(1)}\right) = 			
		\left\langle 
		\raisebox{-10pt}{\begin{tikzpicture}[
					elt/.style={circle,draw=black!100,thick, inner sep=0pt,minimum size=1.3ex},scale=0.8]
				\begin{scope}[xshift=0cm,yshift=0cm]
				\path 	(0,0) 	node 	(b0) [elt, label={[xshift=0pt, yshift = -20 pt] $\beta_{0}$} ] {}
				        (1,0) node 	(b1) [elt, label={[xshift=0pt, yshift = -20 pt] $\beta_{1}$} ] {};
				\draw [black,thick, double distance = .4ex] (b0) -- (b1);
				\draw [purple, dashed, line width = 0.5pt] (0.5,-0.3) -- (0.5,0.3)	node [pos=1,above] {\small $\sigma_{1}\sigma_{2}\sigma_{1}\sigma_{2}$};
				\end{scope}
				\begin{scope}[xshift=2cm,yshift=0cm]
				\path 	(0,0) 	node 	(b0) [elt, label={[xshift=0pt, yshift = -20 pt] $\gamma_{0}$} ] {}
				        (2,0) node 	(b1) [elt, label={[xshift=0pt, yshift = -20 pt] $\gamma_{1}$} ] {};
				\draw [black,thick, double distance = .4ex] (b0) -- (b1);
				\draw [purple, dashed, line width = 0.5pt] (1,-0.3) -- (1,0.3)	node [pos=1,above] {\small $\sigma_{1}$};
				\end{scope}
			\end{tikzpicture}}
		 \right\rangle,		
	\end{equation}
	where 
	\begin{equation}
		\beta_{0} = \alpha_{0} + \alpha_{3}, \quad \beta_{1} = \alpha_{1} + \alpha_{2}, \quad 
		\gamma_{0} = \alpha_{0} + \alpha_{2},\quad \gamma_{1} = \alpha_{1} + \alpha_{3}.
	\end{equation}
\end{theorem}

We will give two proofs of this below, the first by direct calculation using elementary facts about affine Weyl groups and the second making use of the theory of normalizers of parabolic subgroups in Coxeter groups due to Brink and Howlett \cite{howlett,brinkhowlett} (see \cite{yangtranslations} for an introduction to applications of this theory in the context of discrete Painlev\'e equations). 

Before this let us make some remarks about the description of $G$ as an affine Weyl group and also its realization on the level of $\operatorname{Pic}(\mathcal{X})$.
	Note that $\beta_{i}$ are roots of the $A_3^{(1)}$ root system in $V^{(1)}$ since $\beta_{i}\bullet \beta_{i} = -2$, 
	and the corresponding reflections as automorphisms of $\operatorname{Pic}(\mathcal{X})$ are defined in the usual way,
	\begin{equation*}
		r_{\beta_{0}} = r_{\alpha_{0} + \alpha_{3}} = w_{0} w_{3} w_{0},\qquad r_{\beta_{1}} = r_{\alpha_{1} + \alpha_{2}} = w_{1} w_{2} w_{1}.
	\end{equation*}
	However, 
	$\gamma_{i}\bullet \gamma_{i} = -4$ and the reflection
	$r_{\gamma_{i}}(C) = C + \frac{1}{2} (\gamma_{i}\bullet \mathcal{C}) \gamma_{i}$ is only defined on 
	$\operatorname{Pic}^{\mathbb{R}}(\mathcal{X})$. 
	Nevertheless, this formula does define the action on the $\alpha$-roots, since
	\begin{equation*}
		r_{\gamma_{0}}(\alpha_{2k}) = \alpha_{2k} - \gamma_{0},\quad r_{\gamma_{0}}(\alpha_{2k+1}) = \alpha_{2k+1} + \gamma_{0},\quad 
		r_{\gamma_{1}}(\alpha_{2k}) = \alpha_{2k} + \gamma_{1},\quad r_{\gamma_{1}}(\alpha_{2k+1}) = \alpha_{2k+1} - \gamma_{1}.
	\end{equation*}
	Thus, we can decompose this action in terms of the generators,
	\begin{equation}\label{eq:gamma-root-action}
		r_{\gamma_{0}}( \alpha_{0}, \alpha_{1}, \alpha_{2}, \alpha_{3}  ) 
		= ( - \alpha_{2}, \alpha_{0} + \alpha_{1} + \alpha_{2}, -\alpha_{0}, \alpha_{0} + \alpha_{2} + \alpha_{3} )
		= (\sigma_{2} w_{2} w_{0}) (  \alpha_{0}, \alpha_{1}, \alpha_{2}, \alpha_{3} ), 
	\end{equation}
	and so we interpret $r_{\gamma_{0}}$ as $r_{\gamma_{0}} = \sigma_{2} w_{2} w_{0}$ on the whole of $\operatorname{Pic}(\mathcal{X})$. Similarly, $r_{\gamma_{1}} = \sigma_{3} w_{3} w_{1}$.
	
	Note that the induced action of the dynamics \eqref{eq:dP-II-sys} on the new roots becomes a translation on the $\gamma$-sub-lattice
	\begin{equation*}
		\varphi_{*}(  \beta_{0}, \beta_{1}  ) = 
        (\beta_{0}, \beta_{1}) ,\qquad 
		\varphi_{*}(  \gamma_{0}, \gamma_{1} ) =  (\gamma_{0} - \delta, \gamma_{1} + \delta )
		= (\gamma_{0}, \gamma_{1} ) + (-1,1)  \delta,
	\end{equation*}
	and so formally, $\varphi_{*} = \sigma_{1} r_{\gamma_{0}}$. Indeed, with the interpretation \eqref{eq:gamma-root-action}, this is exactly what we have,
	\begin{equation*}
		\varphi_{*} = \sigma_{1} r_{\gamma_{0}} = \sigma_{1} \sigma_{2} w_{2} w_{0}.
	\end{equation*}
    The fact that the element $\varphi_*$, which is a quasi-translation in $\widehat{W}\left(A_3^{(1)}\right)$, is a translation element with respect to the structure of $G$, as an extended affine Weyl group, can be understood in terms of normalizer theory. This will be explained in Remark \ref{rem:translations} at the end of this section. 

\begin{proof}[Proof of Theorem \ref{th:maintheorem} (direct computation)]

Recall from section \ref{sub:affine_weyl_symmetry_group} that we have 
\begin{equation} \label{extendedisoms}
\widehat{W}\left(A_3^{(1)}\right) > \widetilde{W}\left(A_3^{(1)}\right)\cong  W\left(A_3^{(1)}\right)\rtimes\Sigma \cong W\left(A_3\right)\ltimes \mathbf{T}_{P} \cong \mathbf{T}_{P} \rtimes W\left(A_3\right),
\end{equation}
and isomorphisms \eqref{extendedisoms} are realized by 
\begin{equation}
    \begin{aligned}
        \rho   = \sigma_{1} \sigma_2 &= w_{1} w_{2} w_{3} \mathbf{T}_{h_3} = \mathbf{T}_{-h_1} w_{1} w_{2} w_{3},\\
        \rho^2 =  \sigma_{1} \sigma_2  \sigma_{1} \sigma_2 &= w_{2} w_{3} w_{1} w_2 \mathbf{T}_{h_2} = \mathbf{T}_{-h_2} w_{2} w_{3} w_{1} w_{2},\\
        \rho^3 =  \sigma_{2} \sigma_{1} &= w_{3} w_{2} w_{1} \mathbf{T}_{h_{1}} = \mathbf{T}_{-h_3} w_{3} w_{2} w_{1}.
    \end{aligned}
\end{equation}
We have $\left(\operatorname{Aut}\left(A_3^{(1)}\right)/\Sigma\right) \cong \mathbb{Z}_2$ so $\widetilde{W}\left(A_3^{(1)}\right)$ is a normal subgroup of $\widehat{W}\left(A_3^{(1)}\right)$, and we can choose $\sigma_1$ as a representative of the nontrivial coset. 
Then we can write any element of $\widehat{W}\left(A_3^{(1)}\right)$ as
\begin{equation}
    \mathbf{T}_{h} w \sigma,\qquad h\in P, \quad w \in W(A_3), \quad \sigma \in \{1,\sigma_1\}.
\end{equation}
The idea then is, for each of the finite number of choices of $\sigma$ and $w$ to represent the cosets $\mathbf{T}_{P} w \sigma \in  \mathbf{T}_{P}\backslash \widehat{W}(A_3^{(1)}) $, with $\sigma \in \operatorname{Aut}(A_3^{(1)})$, and $w\in W(A_3)$, we can compute $\left\{ h\in P~|~  \mathbf{T}_{h} w \sigma (\alpha_2+\alpha_3)=\alpha_2+\alpha_3\right\}$, and similarly $\left\{ h\in P~|~ \mathbf{T}_{h} w \sigma (\alpha_2+\alpha_3)=\alpha_0+\alpha_1\right\}$.
Since $W(A_3)$ is of order 24, we have 48 choices for $w \sigma$. 
We list the infinite families of elements in the form $\mathbf{T}_h w \sigma$ that exhausts all elements of the stabilizer in Figure~\ref{fig:transl}.

\begin{figure}[ht]
\centering
\begin{equation*}
\begin{array}{c|c|c|c}
    \mathbf{T}_h  & w & \sigma & \mathbf{T}_{h} w \sigma (\alpha_2+\alpha_3) \\[5pt]
 \hline
     \rule{0pt}{2.9ex}\mathbf{T}_{h_1}^{\ell_1} \mathbf{T}_{h_2}^{\ell_2}\mathbf{T}_{h_3}^{-\ell_2}, ~~\ell_1, \ell_2\in \mathbb{Z} & e & e & \alpha_2+\alpha_3 \\[3pt]
      \mathbf{T}_{h_1}^{\ell_1} \mathbf{T}_{h_2}^{\ell_2}\mathbf{T}_{h_3}^{-\ell_2}, ~~\ell_1, \ell_2\in \mathbb{Z} & w_{1} w_{2} w_{1} & e & \alpha_2+\alpha_3 \\[3pt]
      \mathbf{T}_{h_1}^{\ell_1} \mathbf{T}_{h_2}^{\ell_2}\mathbf{T}_{h_3}^{-\ell_2-1}, ~~\ell_1, \ell_2\in \mathbb{Z} & w_{3} w_{2} w_{3} & \sigma_1 & \alpha_2+\alpha_3 \\[3pt]
      \mathbf{T}_{h_1}^{\ell_1} \mathbf{T}_{h_2}^{\ell_2}\mathbf{T}_{h_3}^{-\ell_2-1}, ~~\ell_1, \ell_2\in \mathbb{Z} & w_{3} w_{2} w_{3} w_{1} w_{2} w_{1}& \sigma_1 & \alpha_2+\alpha_3 \\[3pt]
      \mathbf{T}_{h_1}^{\ell_1} \mathbf{T}_{h_2}^{\ell_2}\mathbf{T}_{h_3}^{-\ell_2-1}, ~~\ell_1, \ell_2\in \mathbb{Z} & w_{3} w_{2} w_{3}& e & \alpha_0+\alpha_1 \\[3pt]
      \mathbf{T}_{h_1}^{\ell_1} \mathbf{T}_{h_2}^{\ell_2}\mathbf{T}_{h_3}^{-\ell_2-1}, ~~\ell_1, \ell_2\in \mathbb{Z} & w_{3} w_{2} w_{3} w_{1} w_{2} w_{1}& e & \alpha_0+\alpha_1 \\[3pt]
      \mathbf{T}_{h_1}^{\ell_1} \mathbf{T}_{h_2}^{\ell_2}\mathbf{T}_{h_3}^{-\ell_2}, ~~\ell_1, \ell_2\in \mathbb{Z} & e& \sigma_1 & \alpha_0+\alpha_1 \\[3pt]
      \mathbf{T}_{h_1}^{\ell_1} \mathbf{T}_{h_2}^{\ell_2}\mathbf{T}_{h_3}^{-\ell_2}, ~~\ell_1, \ell_2\in \mathbb{Z} & w_{1} w_{2} w_{1}& \sigma_1 & \alpha_0+\alpha_1
    \end{array}
\end{equation*}	
\caption{Elements in $G$ and their action on $\alpha_{2} + \alpha_{3}$}
	\label{fig:transl}
\end{figure}

This allows us to see that the stabilizer is generated as follows:
\begin{equation} \label{stabilizergenset}
    \operatorname{Stab}\left\{\alpha_0+\alpha_1,\alpha_2+\alpha_3\right\} = \langle w_{1} w_{2} w_{1}, \sigma_1, \mathbf{T}_{h_1}, \mathbf{T}_{h_2-h_3}, \mathbf{T}_{-h_3} w_{3} w_{2} w_{3}\rangle.
\end{equation} 
Computing some relations among these we can identify this by inspection as isomorphic to
\begin{equation}
\widetilde{W}(A_1^{(1)})\times \widetilde{W}(A_1^{(1)}) 
= \left\langle r_{\beta_0},r_{\beta_1},\pi_{\beta} ~|~
\begin{gathered}
    r_{\beta_i}^2 =\pi_{\beta}^2= e, \\
    \pi_{\beta} r_{\beta_{0}}=r_{\beta_{1}}\pi_{\beta}
\end{gathered} 
\right\rangle 
\times  
\left\langle r_{\gamma_0},r_{\gamma_1},\pi_{\gamma} ~|~
\begin{gathered}
    r_{\gamma_i}^2 =\pi_{\gamma}^2= e, \\
    \pi_{\gamma} r_{\gamma_{0}}=r_{\gamma_{1}}\pi_{\gamma}
\end{gathered} 
\right\rangle.
\end{equation}
This is via the following expressions for the generators of $\widetilde{W}\left(A_1^{(1)}\right)\times \widetilde{W}\left(A_1^{(1)}\right)$ in terms of those of the stabilizer as in \eqref{stabilizergenset}.
\begin{equation}
\begin{gathered}
    r_{\beta_0} = w_{1}w_{2}w_{1} \mathbf{T}_{h_1+h_2-h_3}, \quad    
    r_{\beta_1} = w_{1} w_{2} w_{1}, \quad
    \pi_{\beta} = w_{1}w_{2}w_{1}\mathbf{T}_{-h_3} w_{3}w_{2}w_{3} \mathbf{T}_{h_1}, \\
    r_{\gamma_0} = \mathbf{T}_{-h_3} w_{3}w_{2} w_{3} \sigma_1, \quad 
    r_{\gamma_1} = \sigma_1 \mathbf{T}_{-h_3} w_{3}w_{2} w_{3}, \quad
    \pi_{\gamma} = \sigma_1.
\end{gathered}
\end{equation}
We can also rewrite these in terms of the original generators $w_0,\dots,w_3,\sigma_1,\sigma_2$ of $\widehat{W}\left(A_3^{(1)}\right)$ as follows:
\begin{equation}
    \begin{gathered}
        r_{\beta_0} = w_0 w_3 w_0, \quad r_{\beta_1} = w_{1} w_{2} w_{1}, \quad \pi_{\beta} = (\sigma_1 \sigma_2)^2, \\
        r_{\gamma_0} = \sigma_2 w_2 w_0, \quad r_{\gamma_1} = \sigma_1 \sigma_2 \sigma_1 w_3 w_1, \quad \pi_{\gamma} = \sigma_1.
    \end{gathered}
\end{equation}
Verifying that the subgroups $\langle r_{\beta_0},r_{\beta_1},\pi_{\beta},r_{\gamma_0},r_{\gamma_1},\pi_{\gamma}\rangle$ and $\langle w_1 w_2 w_1,\sigma_1,\mathbf{T}_{h_1},\mathbf{T}_{h_{2}-h_{3}},\mathbf{T}_{-h_3}w_3w_2w_3\rangle$ coincide is done by direct calculation.

\end{proof}

Before we give the proof of Theorem \ref{th:maintheorem} using normalizer theory, note that this approach is motivated by the following earlier observations. 
The element $\varphi_{*}$ that gives rise to
the d-$\Pain{II}$ equation \eqref{eq:dP-II-fs} is a
quasi-translation of order two by squared length one, as defined in \cite{Shi:2019:VAAATDPE}. 
That is, $\varphi_{*}^{2}=\psi_{*}$, where $\psi_{*}$ is a translation 
associated to a weight of squared length one. 
Moreover, it was found in Section~\ref{sub:the_symmetry_group_of_the} that the problem of finding the symmetries
of the d-$\Pain{II}$ equation is reduced to finding the setwise stabilizer of $\left\{\alpha_0+\alpha_1, \alpha_2+\alpha_3\right\}\cong A_1^{(1)}$
in $\widehat{W}\left(A_3^{(1)}\right)$. 
The stabilizer of an $A_1^{(1)}$
subsystem in an affine Weyl group (or an extension thereof) can be computed largely by methods developed to compute the normalizer of a standard parabolic subgroup of a Coxeter group. We can make use of general results of \cite{brinkhowlett} to compute a set of generators, establish its structure as an extended affine Weyl group with an underlying root system, and then construct an element of quasi-translation of order two by squared length one 
from considering a translation in the weight lattice of this underlying root system. Finally we show that this quasi-translation is related
to the element $\varphi_{*}$ by conjugation. That is, the symmetry group of $\varphi_{*}$ is this extended affine Weyl group under the same conjugation.

\begin{proof}[Proof of Theorem \ref{th:maintheorem} (normalizer theory)]

When $W$ is an affine Weyl group with simple system $\Delta^{(1)}$ and root system $\Phi=W\cdot \Delta^{(1)}$, a subset $J \subset \Delta^{(1)}$ defines a \emph{standard parabolic subgroup} $W_J = \langle w_{\alpha} ~|~ \alpha \in J\rangle$. 
If $W_J$ can be written as a product $W_J = W_{I_1} \times W_{I_2}\times \cdots \times W_{I_N}$ of standard parabolic subgroups with $I_i\neq \emptyset$ then we call $I_1,\dots,I_N$ the \emph{irreducible components} of $J$. 
The standard parabolic subgroup $W_J$ is itself a Coxeter group and we can consider the associated root system $\Phi_J=W_J\cdot J\subset\Phi$.  
When $W_J$ is finite, the subset $J$ is called \emph{spherical}, and in this case the irreducible components of $W_J$ will be finite Weyl groups.
The results of \cite{brinkhowlett} that we will use state that the normalizer $N(W_J)$ of $W_J$ in $W$ is given by $N(W_J) = N_J \ltimes W_J$, where $N_J=\left\{ w \in W ~|~ w J = J\right\}$ is the setwise stabiliser of $J$, 
and, further, provide a way to obtain a presentation of $N_J$ in terms of generators and relations.

To outline the Brink-Howlett construction of the presentation of $N_J$, we introduce some notation. 
For a spherical subset $I\subset \Delta^{(1)}$ we let $w_I\in W_I$ denote the unique longest element in $W_I$ with respect to the usual Bruhat ordering on $W$. This is the ordering of elements by length $\ell(w)$, defined as the minimal length of an expression for $w$ in terms of simple reflections.
When $I=\emptyset$ we regard $w_I$ as the identity element.
In the Brink-Howlett framework, disjoint subsets $I,J\subset \Delta^{(1)}$ determine a unique element $v[I,J]$, which in the cases relevant to us has the following simple description.
Let $I,J\subset \Delta^{(1)}$ with $I\cap J=\emptyset$ be such that $\Phi_{I\cup J} \setminus \Phi_J$ is finite.
Let $L=I\cup J$ and denote by $L_0$ the union of the irreducible components of $L$, which have nonempty intersection with $I$.
Then $L_0$ is spherical and
\begin{equation}
    v[I,J]= w_{L_0} w_{L_0\cap J}.
\end{equation}
 When $I=\{\alpha\}$ consists of only a single simple root we write $v[\alpha,J]$, and note that when $\alpha$ corresponds to a node of the Dynkin diagram not joined to any of those corresponding to elements of $J$ we have $v[\alpha,J]=w_{\alpha}$ since $L_0\cap J=\emptyset$ and $L_0=I$.

 For an affine Weyl group $W$ and $J\subset \Delta^{(1)}$, the group presentation of $N_J$ is obtained in the Brink-Howlett framework using a groupoid constructed as follows. 
 Let $\mathcal{J}=\left\{ K \subset \Delta^{(1)} ~|~ K=w J \text{ for some } w\in W\right\}$ be the set of $W$\emph{-associates} of $J$.
 Then consider the set
 \begin{equation*}
     \mathcal{G}_{\mathcal{J}} = \left\{ (J_0,v[\alpha_{i_{n-1}},J_{n-1}]\cdots v[\alpha_{i_1},J_1]v[\alpha_{i_0},J_0],J_n)\ | \ J_1,\dots,J_n \in \mathcal{J},\ \alpha_{i_m}\in\Delta^{(1)}, \ v[\alpha_{i_m}] J_{m} = J_{m+1}\right\}.
 \end{equation*}
This has the structure of a groupoid with (partial) operation $(L,v_2,I)(K,v_1,L) = (K,v_2v_1,I)$. 
This can be represented as a graph whose vertices correspond to elements of $\mathcal{J}$, with vertices $I,K$ connected by an edge when there exists $\alpha\in\Delta^{(1)}$ such that $v[\alpha,I]I=K$. 
The graph may also have \emph{loops}, i.e. edges from a vertex $I\in \mathcal{J}$ to itself, when there exists $\alpha\in \Delta^{(1)}$ such that $v[\alpha,I]I=I$.

Then paths in the graph correspond to elements of $\mathcal{G}_{\mathcal{J}}$, and elements of $N_J$ correspond to paths beginning and ending at the vertex $J$:
\begin{equation}
    N_J = \left\{ \boldsymbol{v}=v[\alpha_{i_{n-1}},J_{n-1}]\cdots v[\alpha_{i_1},J_1]v[\alpha_{i_0},J_0] \quad | \quad (J_0, \boldsymbol{v}, J_n)\in \mathcal{G}_{\mathcal{J}}, \quad J_0=J_n =J \right\}.
\end{equation}
Further, $N_J$ itself has the structure of a Coxeter group (see \cite{brinkhowlett} for details), which often turns out to be an affine Weyl group or an extension thereof. 
The above results do not immediately solve our problem of computing and describing $G$, but they take care of most of it.
Some extra work is required, which comes from, firstly, the fact that $\widehat{W}\left(A_3^{(1)}\right)$ is not purely an affine Weyl group and, secondly, the fact that we want to compute the setwise stabilizer not of some $J \subset \Delta^{(1)}$, but rather of $J\cup \{\delta - \theta\}$, where $J$ is spherical and $\theta$ is the highest root of $\Phi_J$.
With the first fact in mind, we begin by noting that we can view $W\left(A_3^{(1)}\right) \rtimes \langle \sigma_2 \rangle < \widetilde{W}\left(A_3^{(1)}\right)< \widehat{W}\left(A_3^{(1)}\right)$ as the extended affine Weyl group $\widetilde{W}\left(B_3^{(1)}\right) = W\left(B_3^{(1)}\right)\rtimes \operatorname{Aut}\left(B_3^{(1)}\right)$.
To do this, take the simple system $\Delta^{(1)}$ of the $A_3^{(1)}$ root system in $V^{(1)}$ as above and introduce 
\begin{equation}\label{AtoB}
    \widetilde{\Delta}^{(1)}=(\mathfrak{b}_0,\mathfrak{b}_1,\mathfrak{b}_2,\mathfrak{b}_3), \qquad \text{where}\quad \mathfrak{b}_0 = \alpha_0, \quad \mathfrak{b}_1 = \alpha_2, \quad \mathfrak{b}_2 = \alpha_3, \quad \mathfrak{b}_3 = \frac{\alpha_1-\alpha_3}{2}.
\end{equation}
In particular, from Equations \eqref{AtoB} and
\eqref{CA31}, we see that $(\mathfrak{b}_j,\mathfrak{b}_j)=2$, $j=0,1,2$ and
$(\mathfrak{b}_3,\mathfrak{b}_3)=1$.
Then $\widetilde{\Delta}^{(1)}$ forms a simple system for a root system of type $B_3^{(1)}$, with Cartan matrix given by 
\begin{equation}
\begin{aligned}
C(B_3^{(1)})=\left(2\frac{(\mathfrak{b}_i,\mathfrak{b}_j)}{(\mathfrak{b}_j,\mathfrak{b}_j)}\right)_{i,j=0}^{3}
&=\left(
\begin{array}{cccc}
 2 & 0 & -1 & 0 \\
 0 & 2 & -1 & 0 \\
 -1 & -1 & 2 & -2 \\
 0 & 0 & -1 & 2 \\
\end{array}
\right),\; 
\mbox{where} \\
    \left((\mathfrak{b}_i,\mathfrak{b}_j)\right)_{i,j=0}^{3}
    &= 
    \left(
\begin{array}{cccc}
 2 & 0 & -1 & 0 \\
 0 & 2 & -1 & 0 \\
 -1 & -1 & 2 & -1 \\
 0 & 0 & -1 & 1 \\
\end{array}
\right).
\end{aligned}
\end{equation}
Here $(~\,,~)$ is the same bilinear form  on $V^{(1)}$ as was used for the $A_3^{(1)}$ root system, but note that for $B$-type root systems the generalised Cartan matrix is not symmetric.
The affine Weyl group 
\begin{equation*}
	W\left(B_{3}^{(1)}\right) = W\left(\raisebox{-20pt}{\begin{tikzpicture}[scale=.55,
			elt/.style={circle,draw=black!100,thick, inner sep=0pt,minimum size=1.6mm}]
		\path 	(0,1) 	node 	(b3) [elt, label={[xshift=-10pt, yshift = -10 pt] \small $\mathfrak{b}_{3}$} ] {}
		        (0,-.16) node 	(b2) [elt, label={[xshift=-10pt, yshift = -10 pt]\small  $\mathfrak{b}_{2}$} ] {}
		        ( 1,-1) node  	(b1) [elt, label={[xshift=10pt, yshift = -10 pt] \small $\mathfrak{b}_{1}$} ] {}
		        ( -1,-1) 	node 	(b0) [elt, label={[xshift=-10pt, yshift = -10 pt] \small $\mathfrak{b}_{0}$} ] {};
		\draw [black,line width=.9pt ] (.12,.92) -- (.12,-.07); 
		\draw [black,line width=.9pt ] (-.12,.92) -- (-.12,-.07); 
            \draw [black, line width=.5pt] (0,.5) -- (.2,.3);
            \draw [black, line width=.5pt] (0,.5) -- (-.2,.3);
		\draw [black,line width=1pt ] (b2) -- (b1); 
		\draw [black,line width=1pt ] (b2) -- (b0) ; 
	\end{tikzpicture}} \right)
	=
	\left\langle s_{0},\dots, s_{3}\ |\ 
	s_{i}^{2} = e,\quad  
    \begin{alignedat}{2}
    (s_{i} s_{j})^2 &= e& &\text{ when 
   				\raisebox{-0.08in}{\begin{tikzpicture}[
   							elt/.style={circle,draw=black!100,thick, inner sep=0pt,minimum size=1.5mm}]
   						\path   ( 0,0) 	node  	(ai) [elt] {}
   						        ( 0.5,0) 	node  	(aj) [elt] {};
   						\draw [black] (ai)  (aj);
   							\node at ($(ai.south) + (0,-0.2)$) 	{$\mathfrak{b}_{i}$};
   							\node at ($(aj.south) + (0,-0.2)$)  {$\mathfrak{b}_{j}$};
   							\end{tikzpicture}}}\\
    (s_{i} s_{j})^3  &= e& &\text{ when 
   				\raisebox{-0.13in}{\begin{tikzpicture}[
   							elt/.style={circle,draw=black!100,thick, inner sep=0pt,minimum size=1.5mm}]
   						\path   ( 0,0) 	node  	(ai) [elt] {}
   						        ( 0.5,0) 	node  	(aj) [elt] {};
   						\draw [black, thick] (ai) -- (aj);
   							\node at ($(ai.south) + (0,-0.2)$) 	{$\mathfrak{b}_{i}$};
   							\node at ($(aj.south) + (0,-0.2)$)  {$\mathfrak{b}_{j}$};
   							\end{tikzpicture}}}\\
    (s_{i} s_{j})^4  &= e& &\text{ when 
   				\raisebox{-0.17in}{\begin{tikzpicture}[
   							elt/.style={circle,draw=black!100,thick, inner sep=0pt,minimum size=1.5mm}]
   						\path   ( 0,0) 	node  	(ai) [elt] {}
   						        ( 0.5,0) 	node  	(aj) [elt] {};
   						\draw [black, thick] ( 0.035,-0.06) -- ( 0.465,-.06);
   						\draw [black, thick] ( 0.035,0.06) -- ( 0.465,0.06);
   							\node at ($(ai.south) + (0,-0.2)$) 	{$\mathfrak{b}_{i}$};
   							\node at ($(aj.south) + (0,-0.2)$)  {$\mathfrak{b}_{j}$};
   							\end{tikzpicture}}}
	\end{alignedat}\right\rangle,
\end{equation*} 
is realised as linear transformations of $V^{(1)}$ by the actions of simple reflections
    $s_i(v) = v - 2\frac{(v,\mathfrak{b}_i)}{(\mathfrak{b}_i,\mathfrak{b}_i)} \mathfrak{b}_i.$
The root system here is then $\widetilde{\Phi}^{(1)}=W\left(B_3^{(1)}\right)\cdot\widetilde{\Delta}^{(1)}$, and $W\left(B_3^{(1)}\right)$ includes reflections $s_{\mathfrak{b}}$, $\mathfrak{b}\in \widetilde{\Phi}^{(1)}$ which act on $V^{(1)}$ by
\begin{equation}
    s_{\mathfrak{b}}(v) = v - 2\frac{(v,\mathfrak{b})}{(\mathfrak{b},\mathfrak{b})} \mathfrak{b}.
\end{equation}
Note that $\alpha_1 = \mathfrak{b}_2 + 2 \mathfrak{b}_3\in \widetilde{\Phi}^{(1)}$ is a root of the $\mathfrak{b}$-system, so $W\left(A_3^{(1)}\right)\subset W\left(B_3^{(1)}\right)$,
but that not all roots of the $\mathfrak{b}$-system are roots of the $\alpha$-system, and in particular the simple reflection $s_3$ acts on $V^{(1)}$ as the automorphism $\sigma_1\sigma_2\sigma_1$ of the $A_3^{(1)}$ Dynkin diagram of the $\alpha$-system.
This is why we use notation $s_{\mathfrak{b}}$ to distinguish the reflections associated to elements of $\widetilde{\Phi}^{(1)}$ from the reflections $r_{\alpha}\in W\left(A_3^{(1)}\right)$.
The single non-trivial automorphism of the $B_3^{(1)}$ Dynkin diagram is given by $\sigma_2$ and we have
$W\left(A_3^{(1)}\right)\rtimes \langle \sigma_2, \sigma_1\sigma_2\sigma_1\rangle = \widetilde{W}\left( B_3^{(1)}\right)=\widehat{W}\left( B_3^{(1)}\right)=\langle s_0,s_1,s_2,s_3,\sigma_2\rangle$, with generators satisfying the defining relations given in Equations
\eqref{funWB3} and \eqref{funWaB3} as summarised on Figure \ref{fig:B3system}:
\begin{subequations}
    \begin{gather}
        s_j^2=1, \;(j\in \{0,1,2,3\}),\nonumber \\
        (s_1s_{2})^3=(s_1s_{3})^2=(s_2s_{3})^4=(s_0s_{2})^3=(s_0s_{3})^2=(s_{0}s_1)^2=e,\label{funWB3}\\\label{funWaB3}
         \sigma_2^2=e, \;\;\sigma_2s_0=s_1\sigma_2.
    \end{gather}
\end{subequations}
The generator $\sigma_1$ of $\operatorname{Aut}(A_3^{(1)})$ is not accounted for in this extension, and we have $\widetilde{W}\left(A_3^{(1)}\right) = \widetilde{W}\left(B_3^{(1)}\right)\rtimes \langle \sigma_1 \rangle$.
The reason for introducing this description is that it makes it easier to apply the normalizer theory to compute the relevant stabilizer -- we can now work with an extension of an affine Weyl group by a smaller Dynkin diagram automorphism group.
In addition, the different root lengths of the two $A_1^{(1)}$ root systems in Theorem \ref{th:maintheorem} can be clearly seen in terms of roots of different lengths in the $B_3^{(1)}$ system.

\begin{figure}[ht]
\begin{equation}\label{eq:a-roots-b31}			
	\raisebox{-42.1pt}{\begin{tikzpicture}[
			elt/.style={circle,draw=black!100,thick, inner sep=0pt,minimum size=2mm}]
		\path 	(0,1) 	node 	(b3) [elt, label={[xshift=-10pt, yshift = -10 pt] $\mathfrak{b}_{3}$} ] {}
		        (0,-.16) node 	(b2) [elt, label={[xshift=-10pt, yshift = -10 pt] $\mathfrak{b}_{2}$} ] {}
		        ( 1,-1) node  	(b1) [elt, label={[xshift=10pt, yshift = -10 pt] $\mathfrak{b}_{1}$} ] {}
		        ( -1,-1) 	node 	(b0) [elt, label={[xshift=-10pt, yshift = -10 pt] $\mathfrak{b}_{0}$} ] {};
		\draw [black,line width=.9pt ] (.1,1) -- (.1,-.16); 
		\draw [black,line width=.9pt ] (-.1,1) -- (-.1,-.16); 
            \draw [black, line width=.5pt] (0,.5) -- (.2,.3);
            \draw [black, line width=.5pt] (0,.5) -- (-.2,.3);
		\draw [black,line width=1pt ] (b2) -- (b1); 
		\draw [black,line width=1pt ] (b2) -- (b0) ; 
		\draw [purple, dashed, line width = 0.5pt] (0,-1.5) -- (0,1.3)	node [pos=0,below] {\small $\sigma_{2}$};
	\end{tikzpicture}} \qquad
        \begin{gathered}
            \begin{alignedat}{2}
			\mathfrak{b}_0 &= \alpha_0, &\qquad & s_{0} = w_{0},  \\
			\mathfrak{b}_1 &= \alpha_2, &\qquad & s_{1} = w_{2}, \\
			\mathfrak{b}_2 &= \alpha_3, &\qquad & s_{2} = w_{3}, \\
			\mathfrak{b}_3 &= \frac{\alpha_1-\alpha_3}{2}, &\qquad &s_{3} = \sigma_3 = \sigma_1 \sigma_2 \sigma_1, 
		\end{alignedat}
        \\
            \delta = \mathfrak{b}_0 + \mathfrak{b}_1 +2 \mathfrak{b}_2+2\mathfrak{b}_3.
        \end{gathered}
\end{equation}
	\caption{Extension of $W\left(A_3^{(1)}\right)$ to $\widetilde{W}\left(B_3^{(1)}\right)$}
	\label{fig:B3system}	
\end{figure}

Note that we can use an element of $\widehat{W}\left(A_3^{(1)}\right)$, e.g. $w_3=s_2$, to send the set $\{\alpha_0+\alpha_1,\alpha_2+\alpha_3\}$ to $\left\{\alpha_2,\delta-\alpha_2  \right\} = \{\mathfrak{b}_1,\delta - \mathfrak{b}_1\}$, so the group $G$ we wish to compute is conjugate to 
\begin{equation}
    \widetilde{G} = \operatorname{Stab}_{\widehat{W}(A_3^{(1)})} \left\{\mathfrak{b}_1, \delta - \mathfrak{b}_1 \right\} = s_2 G s_2.
\end{equation}
To compute $\widetilde{G}$, we first compute $N_J$ for $J=\{\mathfrak{b}_1\}$ in the affine Weyl group $W\left(B_3^{(1)}\right)$ using the Brink-Howlett method, and then find the elements which exchange $\mathfrak{b}_1$ and $\delta-\mathfrak{b}_1$. 
Once this has been achieved, finally, we consider the extra elements coming from the extension of $W\left(B_3^{(1)}\right)$ to $\widehat{W}\left(A_3^{(1)}\right)=\widetilde{W}\left(B_3^{(1)}\right)\rtimes \langle \sigma_1\rangle$.

We begin by constructing the graph associated with the groupoid $\mathcal{G}_{\mathcal{J}}$ for $W=W\left(B_3^{(1)}\right)$ and $J=\{\mathfrak{b}_1\}$.
The vertices correspond to $W$-associates of $J$, the set of which is 
\begin{equation}
    \mathcal{J}=\left\{ J_0 = \left\{\mathfrak{b}_0\right\}, ~~J_1 = \left\{\mathfrak{b}_1\right\}=J, ~~J_2 = \left\{\mathfrak{b}_2\right\}
    \right\}
\end{equation}
where we note that $\{\mathfrak{b}_3\}$ is not included because $\mathfrak{b}_3$ is a short root. 
To determine where the edges are, one proceeds to calculate, for each vertex $J_l\in\mathcal{J}$, the element $v[\mathfrak{b}_k,J_l]$ for each $\mathfrak{b}_k \in \Delta^{(1)}$. 
If $v[\mathfrak{b}_k,J_l] J_l = J_{m}$ then one draws an edge from $J_l$ to $J_m$, labeled by the element $v[\mathfrak{b}_k,J_l]$.
The graph for the case at hand is given in Figure \ref{fig:groupoidB3}.

\begin{figure}[ht]
    \begin{center}
        \begin{tikzpicture}[
			elt/.style={circle,draw=black!100,fill=black!100,thick, inner sep=0pt,minimum size=2mm}]
		\path 	(-3,0) 	node 	(J0) [elt, label={[xshift=0pt, yshift = -25 pt] $J_{0}$} ] {}
		        (0,0) node 	(J2) [elt, label={[xshift=0pt, yshift = -25 pt] $J_{2}$} ] {}
		        ( 3,0) node  	(J1) [elt, label={[xshift=0pt, yshift = -25 pt] $J_{1}$} ] {};
            \path (J0) edge [loop above] node {$s_1$} (J0);
            \path (J0) edge [loop left] node {$s_3$} (J0);
            \path (J0) edge [->] node[above, midway] {$s_0s_2$} (J2);
            \path (J2) edge [->] node[above, midway] {$s_2s_1$} (J1);
            \path (J2) edge [loop above] node {$s_3s_2s_3$} (J2);
            \path (J1) edge [loop above] node {$s_0$} (J1);
            \path (J1) edge [loop right] node {$s_3$} (J1);
	\end{tikzpicture} 
    \end{center}
	\caption{Graph for the groupoid $\mathcal{G}_{\mathcal{J}}$ for $W=W\left(B_3^{(1)}\right)$ and $J=\{\mathfrak{b}_1\}$.}
	\label{fig:groupoidB3}	
\end{figure}
It will be convenient to denote paths in the graph according to the vertices through which they pass,
read from right to left, e.g. $J_2J_0$ for the path from $J_0$ to $J_2$.
To account for loops, we use the symbol for the element itself, e.g. $J_0s_1J_0$ to indicate the loop from $J_0$ to itself via $s_1$.
For a path $p$ in this notation, we denote by $v[p]$ the composition of the elements $v[\mathfrak{b}_i,I_i]$ corresponding to the edges and loops in the path, e.g. $v[J_2J_0s_1J_0]=s_0s_2s_1$ for path that starts at $J_0$, traverses a loop via $s_1$, then goes to $J_2$ via $v[\mathfrak{b}_2,J_0]=s_{0}s_{2}$.

From Figure \ref{fig:groupoidB3}, we see that the elements $w\in N_J$ corresponding to paths starting and ending at $J_1=J$ can be obtained as compositions of the following elements and their inverses:
\begin{equation}
\begin{aligned}
    v[J_1s_0J_1]&=s_0 = s_{\mathfrak{b}_0}, \\
    v[J_1s_3J_1]&=s_3 = s_{\mathfrak{b}_3}, \\
    v[J_1J_2s_3s_2s_3J_2J_1]&= v[J_1 J_2] s_3s_2s_3 v[J_2J_1] = s_2s_1s_3s_2s_3s_1s_2 = s_{\delta-\mathfrak{b}_0}, \\
    v[J_1J_2J_0s_3J_0J_2J_1]&= v[J_1 J_2] v[J_2 J_0] s_3  v[J_0J_2] v[J_2 J_1] = s_2s_1 s_0s_2 s_3 s_2s_0 s_1s_2 = s_{\delta-\mathfrak{b}_3}, \\
    v[J_1J_2J_0s_1J_0J_2J_1]&= v[J_1 J_2] v[J_2 J_0] s_1 v[J_0J_2] v[J_2 J_1] =  s_2s_1 s_0s_2 s_1 s_2s_0 s_1s_2 = s_{\mathfrak{b}_0}, 
\end{aligned}
\end{equation}
where we have identified these as reflections $s_{\mathfrak{b}}$ associated to some roots $\mathfrak{b}\in\widetilde{\Phi}^{(1)}$.

Invoking the Brink-Howlett result that $N_J$ consists of elements corresponding to paths starting and ending at $J$, we have a set of generators.  
Since $(\mathfrak{b}_0,\mathfrak{b}_3)=0$, we have 
\begin{equation} \label{eq:NJdescription}
    N_J = \left\langle s_{\mathfrak{b}_0}, s_{\delta-\mathfrak{b}_0}, s_{\mathfrak{b}_3},s_{\delta-\mathfrak{b}_3} \right\rangle = \left\langle s_{\mathfrak{b}_0}, s_{\delta-\mathfrak{b}_0}\right\rangle 
    \times \left\langle s_{\mathfrak{b}_3}, s_{\delta-\mathfrak{b}_3}\right\rangle.
\end{equation}
Each of the subgroups $\left\langle s_{\mathfrak{b}_0}, s_{\delta-\mathfrak{b}_0}\right\rangle$, $\left\langle s_{\mathfrak{b}_3}, s_{\delta-\mathfrak{b}_3}\right\rangle$ are individually isomorphic to copies of $W(A_1^{(1)})$ defined using simple systems $(\mathfrak{b}_0,\delta-\mathfrak{b}_0)$ and $(\mathfrak{b}_3,\delta-\mathfrak{b}_3)$ as groups of linear transformations of the subspaces of $V^{(1)}$ spanned by these. 
We want to describe the quasi-translation elements of $\widetilde{W}(B_3^{(1)})$, such as that corresponding to the d-$\Pain{II}$ dynamics, as translation elements with respect to the affine Weyl group structure of $N_J$. In the case of non-simply laced type of the ambient group, as we have here with $B_3^{(1)}$, this requires some care.

We introduce elements $\eta_0,\eta_1$ and $\omega_0,\omega_1$ of $V^{(1)}$ to play the roles of simple systems for the two copies $\left\langle s_{\mathfrak{b}_0}, s_{\delta-\mathfrak{b}_0}\right\rangle$ and $\left\langle s_{\mathfrak{b}_3}, s_{\delta-\mathfrak{b}_3}\right\rangle$ of $W(A_1^{(1)})$ in $N_J$. 
Since $\mathfrak{b}_0$ is a long root of $\widetilde{\Phi}^{(1)}$ we take $(\eta_0,\eta_1)=(\mathfrak{b}_0,\delta-\mathfrak{b}_0)$, but for the short root $\mathfrak{b}_3$ we instead introduce $\omega_0=2\mathfrak{b}_3$, and $\omega_1=\delta-\omega_0=\delta-2\mathfrak{b}_3$. 
Even though $\omega_0=2\mathfrak{b}_3$ is not an element of $\widetilde{\Phi}^{(1)}$, by abuse of notation we still write $s_{\omega_0}$ for the reflection $s_{\mathfrak{b}_3}$. 
This scaling is necessary for the description of quasi-translation elements of $\widetilde{W}(B_3^{(1)})$ in terms of weights from the affine Weyl group structure of $N_J$ to be compatible with that of translation elements in terms of the weights of the finite $B_3$ system, which we will say more about in Remark \ref{rem:translations} below. 
It also demonstrates the origin of the different root lengths of the $A_1^{(1)}$-type systems in Theorem \ref{th:maintheorem}. 
We summarise the $\eta$- and $\omega$-systems as well as the relations to the $\mathfrak{b}$- and $\alpha$-systems on Figure \ref{fig:A1A1system}.

\begin{figure}[htb]
\begin{equation}\label{eq:a-roots-a11}			
	\raisebox{-22.1pt}{\begin{tikzpicture}[
			elt/.style={circle,draw=black!100,thick, inner sep=0pt,minimum size=2mm}]
		\path 	(-1,0) 	node 	(eta0) [elt, label={[xshift=-10pt, yshift = -10 pt] $\eta_0$} ] {}
		        (0,0) node 	(eta1) [elt, label={[xshift=10pt, yshift = -10 pt] $\eta_{1}$} ] {}
		        (2,0) node  	(omega0) [elt, label={[xshift=-10pt, yshift = -10 pt] $\omega_{0}$} ] {}
		        (4,0) 	node 	(omega1) [elt, label={[xshift=10pt, yshift = -10 pt] $\omega_{1}$} ] {};
		\draw [black,thick, double distance = .4ex] (eta0) -- (eta1);
		\draw [black,thick, double distance = .4ex] (omega0) -- (omega1);
		\draw [purple, dashed, line width = 0.5pt] (-0.5,-.5) -- (-0.5,.5)	node [pos=0,below] {\small $\pi_{\eta}$};
		\draw [purple, dashed, line width = 0.5pt] (3,-.5) -- (3,.5)	node [pos=0,below] {\small $\pi_{\omega}$};
	\end{tikzpicture}} \qquad 
            \begin{aligned}
			\eta_0 &= \mathfrak{b}_0=\alpha_0,  \\
			\eta_1 &= \delta-\mathfrak{b}_0=\alpha_1+\alpha_2+\alpha_3, \\
			\omega_0 &= 2\mathfrak{b}_3=\alpha_1-\alpha_3, \\
			\omega_1 &= \delta-2\mathfrak{b}_3 = \alpha_0+\alpha_2+2\alpha_3,
		\end{aligned}
\end{equation}
	\caption{Root system of type $(A_1+\underset{|\alpha|^2=4}{A_1})^{(1)}$}
	\label{fig:A1A1system}	
\end{figure}

We next consider the setwise stabilizer in $W(B_3^{(1)})$ of not just $J=\{\mathfrak{b}_1\}$, but 
$\{\mathfrak{b}_1, \delta-\mathfrak{b}_1\}$. 
This is straightforward and we need only to add a single generator since $\operatorname{Stab}_{W(B_3^{(1)})} \left\{\mathfrak{b}_1, \delta-\mathfrak{b}_1\right\} = N_J \cup N_J g$, where $g$ is any element such that $g(\mathfrak{b}_1)=\delta-\mathfrak{b}_1$. 
To see this, note that any element of $W(B_3^{(1)})$ leaves $\delta$ invariant so if $g_1,g_2\in W$ are such that $g_1(\mathfrak{b}_1)=g_2(\mathfrak{b}_1)=\delta-\mathfrak{b}_1$, then $g_1^{-1}g_2\in N_J$.
We can find such an element immediately as $s_2s_1s_3s_0s_2$, which acts on the $\mathfrak{b}_i$ as 
\begin{equation}
    g=s_2s_1s_3s_0s_2 : \mathfrak{b}_0 \mapsto \delta-\mathfrak{b}_0, \quad \mathfrak{b}_1 \mapsto \delta-\mathfrak{b}_1, \quad \mathfrak{b}_2 \mapsto -\mathfrak{b}_2-\delta, \quad \mathfrak{b}_3 \mapsto \delta-\mathfrak{b}_3.
\end{equation}
Adding this element as a generator to the description of $N_J$ in \eqref{eq:NJdescription}, we see it causes the addition of a Dynkin diagram automorphism $\pi_{\eta}$ as indicated on Figure \ref{fig:A1A1system}, since $g=s_2s_1s_3s_0s_2= \pi_{\eta}s_{\omega_1}$. 
We then have 
\begin{equation}
    \operatorname{Stab}_{W(B_3^{(1)})}\left\{ \mathfrak{b}_1, \delta-\mathfrak{b}_1\right\} =  \left\langle s_{\eta_0}, s_{\eta_1},\pi_{\eta}\right\rangle 
    \times \left\langle s_{\omega_0}, s_{\omega_1}\right\rangle  \cong \widetilde{W}\left(A_1^{(1)}\right)\times W\left(A_1^{(1)}\right).
\end{equation}
Lastly, we have to consider the elements of $\operatorname{Aut}\left(A_3^{(1)}\right)$ which have not been accounted for as part of $W\left(B_3^{(1)}\right)$.
First consider the extension of $W(B_3^{(1)})$ to $\widetilde{W}(B_3^{(1)})$ by adding the generator $\sigma_2$, corresponding to the automorphism of the $B_3^{(1)}$ Dynkin diagram that permutes the simple roots according to $\sigma_2=(\mathfrak{b}_0\mathfrak{b}_1)$. 
We see that $\sigma_2$ does not add any new elements to the stabilizer beyond $\pi_{\eta}$ already found, since $\sigma_2 = \pi_{\eta}s_2 s_3s_2$.
Finally, considering the extra Dynkin diagram automorphism $\sigma_1$ which is not accounted for in $\widetilde{W}(B_3^{(1)})$, we see that this does indeed add a new generator to the stabilizer. 
We have $\widetilde{W}(B_3^{(1)})\rtimes \langle \sigma_1\rangle=\widetilde{W}(B_3^{(1)})\cup  \widetilde{W}(B_3^{(1)})\sigma_1$, and $\sigma_1(\mathfrak{b}_1)=\mathfrak{b}_2+2\mathfrak{b}_3$, so because $\sigma_1$ is an involution we have $\sigma_1(\mathfrak{b}_2+2\mathfrak{b}_3)=\mathfrak{b}_1$.
To determine if this extension adds any elements to the setwise stabilizer of $\{\mathfrak{b}_1,\delta-\mathfrak{b}_1\}$, we just have to determine whether there exist $w\in W(B_3^{(1)})$ such that $w(\mathfrak{b}_2+2\mathfrak{b}_3)\in \{\mathfrak{b}_1,\delta-\mathfrak{b}_1\}$, in which case we would get a new element $w\sigma_1$ such that $w\sigma_1(\mathfrak{b}_1)\in\{\mathfrak{b}_1,\delta-\mathfrak{b}_1\}$.
This does indeed happen. For example $w=s_0s_2$ with $w(\mathfrak{b}_2+2\mathfrak{b}_3)=\delta-\mathfrak{b}_1$, and we get an element corresponding to the remaining automorphism of the $(A_1+\underset{|\alpha|^2=4}{A_1^{(1)}})^{(1)}$ diagram $\pi_{\omega}=\sigma_1 s_0s_2$.
Note that adding $\pi_{\omega}$ as a generator to $\operatorname{Stab}_{\widetilde{W}(B_3^{(1)})}\left\{ \mathfrak{b}_1, \delta-\mathfrak{b}_1\right\}=\left\langle s_{\eta_0}, s_{\eta_1},\pi_{\eta},s_{\omega_0}, s_{\omega_1}\right\rangle$ accounts for all of the set $\left\{w \sigma_1 ~|~ w \in \widetilde{W}(B_3^{(1)}),~w\sigma_1(\mathfrak{b}_1)\in \left\{\mathfrak{b}_1,\delta-\mathfrak{b}_1\right\}  \right\}$.
This is because if $g_1 \sigma_1$, $g_2\sigma_1$, with $g_1,g_2\in\widetilde{W}(B_3^{(1)})$, are such that $g_1 \sigma_1(\mathfrak{b}_1)=g_2\sigma_1(\mathfrak{b}_1)=\mathfrak{b}_1$, then $g_1 \sigma_1 g_2\sigma_1 \in \operatorname{Stab}_{\widetilde{W}(B_3^{(1)})}\{\mathfrak{b}_1,\delta-\mathfrak{b}_1\}$, and also $\left\langle s_{\eta_0}, s_{\eta_1},\pi_{\eta},s_{\omega_0}, s_{\omega_1}\right\rangle$ already includes all elements of $\widetilde{W}(B_3^{(1)})$ that interchange $\mathfrak{b}_1$ and $\delta-\mathfrak{b}_1$. 
 
We then arrive at the desired description of the setwise stabilizer
\begin{equation}
    \widetilde{G}=\operatorname{Stab}_{\widehat{W}(A_3^{(1)})}\left\{ \alpha_2, \delta-\alpha_2\right\} = \left\langle s_{\eta_0}, s_{\eta_1}, \pi_{\eta}\right\rangle 
    \times \left\langle s_{\omega_0}, s_{\omega_1}, \pi_{\omega}\right\rangle \cong  \widetilde{W}\left(A_1^{(1)}\right)\times \widetilde{W}\left(A_1^{(1)}\right),
\end{equation}
where expressions for the generators in terms of those of $\widetilde{W}(B_3^{(1)})$ and $\widehat{W}(A_3^{(1)})$ are collected below:
\begin{equation}
    \begin{aligned}
        s_{\eta_0} &=  s_{0}= w_{0}, \\
		s_{\eta_1}  &= s_{1}s_{2}s_{3}s_{2}s_{1}s_{2}s_{3}s_{2}s_{1}=w_{3}w_{1}w_{2}w_{1}w_{3}, \\
         \pi_{\eta}&= \sigma_2 s_2 s_3s_2=\sigma_2\sigma_1\sigma_2\sigma_1 w_3 w_1,
    \end{aligned}
    \quad
    \begin{aligned}
		s_{\omega_0} &= s_{3} = \sigma_1\sigma_2\sigma_1, \\
		s_{\omega_1} &= \sigma_2 s_2 s_0 s_1 s_2 = \sigma_2 w_{3}w_{2}w_{0}w_{3}, \\
        \pi_{\omega} &= \sigma_1 s_0s_2= \sigma_1 w_0 w_3.
    \end{aligned}
\end{equation}
Finally, we conjugate the generators of $\widetilde{G}$ by $w_3=s_2$ to obtain the description of $G$ in Theorem \ref{th:maintheorem} via
\begin{equation}
    \begin{aligned}
        s_2 s_{\eta_0} s_2  &= w_3(w_0)w_3=w_0 w_3 w_0 = r_{\beta_0} , \\
		s_2 s_{\eta_1} s_2   &= w_3(w_{3}w_{1}w_{2}w_{1}w_{3})w_3= w_1w_2w_1 = r_{\beta_1} , \\
        s_2 \pi_{\eta} s_2 &= w_3(\sigma_2\sigma_1\sigma_2\sigma_1 w_3 w_1)w_3=(\sigma_1 \sigma_2)^2,
    \end{aligned}
    \quad
    \begin{aligned}
		s_2 s_{\omega_0} s_2 &= w_3(\sigma_1\sigma_2\sigma_1)w_3 =\sigma_1 \sigma_2 \sigma_1 w_1 w_3=r_{\gamma_1}, \\
		s_2 s_{\omega_1} s_2 &= w_3(\sigma_2 w_{3}w_{2}w_{0}w_{3})w_3= \sigma_2 w_2 w_0=r_{\gamma_0}, \\
        s_2 \pi_{\omega} s_2 &= w_3(\sigma_1 w_0 w_3)w_3 = \sigma_1,
    \end{aligned}
\end{equation}
and we are done. 
\end{proof}    
\begin{remark} \label{rem:translations}
    For the $\omega$-system given in Figure~\ref{fig:A1A1system}
 (the underlying root system for $\left\langle s_{\omega_0}, s_{\omega_1}\right\rangle\cong W(A_1^{(1)})$ found earlier), we choose 
 $\{\omega_0,\omega_1\}=\{2\mathfrak{b}_3,\delta-2\mathfrak{b}_3\}$ instead of 
 using $\{\mathfrak{b}_3,\delta-\mathfrak{b}_3\}$, which also gives a set of simple roots for a root system of $A_1^{(1)}$ type, for reasons given as follows. 
For the $A_1^{(1)}$ root system with simple roots $\omega_0,\omega_1$ (where $\omega_0+\omega_1=\delta$) which gives  $\widetilde{W}\bigl(A^{(1)}_1\bigr)=\left\langle s_{\omega_0}, s_{\omega_1},\pi_{\omega}\right\rangle$, the fundamental weight $H_1^{\omega}$ of 
the underlying $A_1$ root system is defined by 
\begin{equation} \label{A1h1}
 (\omega_1,H_1^{\omega})=1, \quad 
(\delta,H_1^{\omega})=0\quad  \mbox{and}\quad \omega^{\vee}_1=\frac{2\omega_1}{(\omega_1,\omega_1)}=2H_1^{\omega},  
\end{equation}
where $\omega_1^{\vee}$ is the simple coroot of the underlying finite $A_1$ root system. 
A ``translation" by $H_1^{\omega}$ in $\widetilde{W}\bigl(A^{(1)}_1\bigr)$
is given by $t_{H_1^{\omega}}=\pi_{\omega}s_{\omega_1}=\sigma_1 s_0s_2\sigma_2 s_2 s_0 s_1 s_2=\sigma_1\sigma_2s_0s_2=\sigma_1\sigma_2 w_0 w_3$,
which is not a translation in $\widehat{W}\bigl(A^{(1)}_3\bigr)$ since $\pi_{\omega}$ not only
permutes the set $\{\omega_0,\omega_1\}$, but it also permutes the set
$\{\mathfrak{b}_1,\delta-\mathfrak{b}_1\}$ by construction. That is, $t_{H_1^{\omega}}$
is a quasi-translation.
Earlier, we mentioned that we want to construct a quasi-translation
of order two by squared length of one. 
This gives us the condition,
\begin{equation}
    |2H_1^{\omega}|^2=1, \quad\mbox{or}\quad |H_1^{\omega}|^2=\frac{1}{4}.
\end{equation}
Choosing $|\omega_1|^2=4$, we have the simple coroot $\omega^{\vee}_1=\omega_1/2$, 
and $|\omega^{\vee}_1|^2=|\omega_1|^2/4=1$. 
That is, we have $|H_1^{\omega}|^2=1/4$.
In fact
\begin{equation}
    H_1^{\omega}=\frac{\omega^{\vee}_1}{2}=\frac{1}{2}(h_1-h_3),
\end{equation}
in terms of the fundamental weights of the $A_3$ root system. 
One can check that $|h_1-h_3|^2=1$, moreover we have $w_3(h_1-h_3)=h_1-h_2+h_3$. 
That is, $t_{H_1^{\omega}}$
is a quasi-translation
of order two by squared length of one. 
Moreover, it is related  
to $\varphi_{*}$ under conjugation by $w_3$:
\begin{equation}
w_3t_{H_1^{\omega}}w_3=w_3\sigma_1\sigma_2 w_0 w_3w_3=\sigma_1\sigma_2w_2w_0=\varphi_{*}.
\end{equation}
\end{remark}


\section{Relation with the discrete Painlev\'e XXXIV equation and recurrence coefficients for Freud unitary ensembles} 
\label{sec:dP-34}
It is possible to further constrain the parameters in the d-$\Pain{II}$ dynamics. In this section we describe one such example that 
recently appeared in studying gap probabilities of Freud unitary ensembles,  \cite{MinWan:2025:DPXHAFPDFRME}
\footnote{In the next few paragraphs we use the notation of that paper to make it easier to follow. Unfortunately it 
clashes with some of our prior notation, so we change it once we get to the discrete Painlev\'e equation in question.}. 
To study the gap probability
that the interval $(-a,a)$  is free of eigenvalues of the Freud unitary ensemble, the authors consider the weight function of the form
$w(x;a) = w_{0}(x)\chi_{(-a,a)^{c}}(x)$, where $w_{0}(x) = e^{-x^{2m}}$, $m\in \mathbb{Z}_{\geq 1}$, 
and $\chi_{I}(x)$ is the usual characteristic function of the interval $I$. 
Then the gap probability $\mathbb{P}(n;a)$ can be expressed in terms 
of Hankel determinants as
\begin{equation*}
	\mathbb{P}(n;a) = \frac{D_{n}(a)}{D_{n}(0)},
\end{equation*}
where
\begin{align*}
	&D_{n}(a) := \det \left(\int_{-\infty}^{\infty} x^{i + j} w(x;a) \, dx\right)_{i,j=0}^{n-1} = \prod_{j=0}^{n-1} h_{j}(a),
	\intertext{and $h_{j}(a)$ is the usual square $L^{2}$-norm of \emph{monic} orthogonal polynomials,}
	&\int_{-\infty}^{\infty} P_{j}(x; a) P_{k}(x;a) w(x;a)\, dx =: \delta_{j,k} h_{j}(a).
\end{align*}
These monic orthogonal polynomials obey a recurrence relation 
\begin{align*}
	x P_{n}(x;a) &= P_{n+1}(x;a) + \beta_{n}(a) P_{n-1}(x;a),
	\intertext{where the recurrence coefficient $\beta_{n}(a)$ satisfies}
	\beta_{n}(a) &= \frac{h_{n}(a)}{h_{n-1}(a)} = \frac{D_{n+1}(a)D_{n-1}(a)}{D_{n}(a)^{2}}.	
\end{align*}
The recurrence coefficients $\beta_{n}(a)$ are the main objects of study in \cite{MinWan:2025:DPXHAFPDFRME}, where it is shown that 
for $m=1,2,3$ they satisfy the equations in the so-called discrete Painlev\'e XXXIV hierarchy defined in \cite{CreJos:1999:DFSTPH}. 
We are interested in the case $m=1$ when the weight $w_{0}(x) = e^{-x^{2}}$ is the usual Gaussian weight. Then the recurrence
coefficients $\beta_{n}(a)$ satisfy the equation 
\begin{equation}\label{eq:MW-dP34}
	a^{2} (2 \beta_{n} - n)^{2} = \beta_{n} (2 \beta_{n-1} + 2 \beta_{n} - 2n + 1)(2 \beta_{n} + 2 \beta_{n+1} - 2n - 1),
\end{equation}
see \cite[Theorem 2.1]{MinWan:2025:DPXHAFPDFRME}. This difference equation becomes the d-$\Pain{XXXIV}$ equation 
\begin{equation}\label{eq:CJ-dP34}
	(w_{n+1} + w_{n} - z_{n+1})(w_{n} + w_{n-1} - z_{n}) = \frac{(2 w_{n} - C_{3} -z_{n})(2 w_{n} + C_{3} - z_{n+1})}{w_{n}},\quad z_{n} = C_{1} + C_{2} n,
\end{equation}
of \cite[(4.2)]{CreJos:1999:DFSTPH}, where $w_{n} = \frac{4}{a^{2}} \beta_{n}$ and the parameters $C_{i}$ are 
\begin{equation}\label{eq:MW-pars}
		C_{1} = - \frac{2}{a^{2}},\quad C_{2} = \frac{4}{a^{2}},\quad  C_{3} = \frac{2}{a^{2}}.
\end{equation}
 
It turns out that \eqref{eq:CJ-dP34} is just a different geometric realization of the same abstract d-$\Pain{II}$ equation \eqref{eq:abs-dP-II}, but the parameter 
values \eqref{eq:MW-pars} result in an appearance of a so-called \emph{nodal curve}, which further restricts the size of the symmetry group. 

To see this, we first rewrite \eqref{eq:CJ-dP34} as a mapping by putting $f:=w_{n}$, $\overline{f}:=g:=w_{n+1}$, etc., to get
\begin{equation}\label{eq:dP-34-map}
	\phi: (f,g)\mapsto \left(\overline{f},\overline{g}\right) = \left(g, \frac{(2g - C_{3} - z_{n+1})(2 g + C_{3} - z_{n+2})}{g(f + g - z_{n+1})} - g + z_{n+2}\right).
\end{equation}
The base points of this mapping, together with its inverse, are 
\begin{equation}\label{eq:dP-34-points}
	\begin{aligned}
			&\pi_{1}\left(\frac{z_{n+1}-C_{3}}{2},\frac{z_{n+1}+C_{3}}{2}\right), \quad  \pi_{2}\left(\frac{z_{n}+C_{3}}{2},\frac{z_{n+2}-C_{3}}{2}\right),\quad 
			\pi_{3}\left(\infty,0),\quad \pi_{4}(0,\infty\right),\\
			&\pi_{5}\left(\infty,\infty\right)\leftarrow \pi_{6}\left(u_{5} = 0, v_{5}= - 1\right)\leftarrow \pi_{7}\left(u_{6}=0,v_{6}=-(z_{n+1}+4)\right)\\
			&\phantom{\pi_{5}\left(\infty,\infty\right)}
			\leftarrow \pi_{8}\left(u_{7}=0,v_{7}=2C_{2} - (z_{n+1}+4)^{2}\right),
	\end{aligned}
\end{equation}
and the resulting point configuration and its resolution are shown on Figure~\ref{fig:dP-34-surface}, where we use $F_{i}$ to denote 
the exceptional divisor of the blowup at $\pi_{i}$.

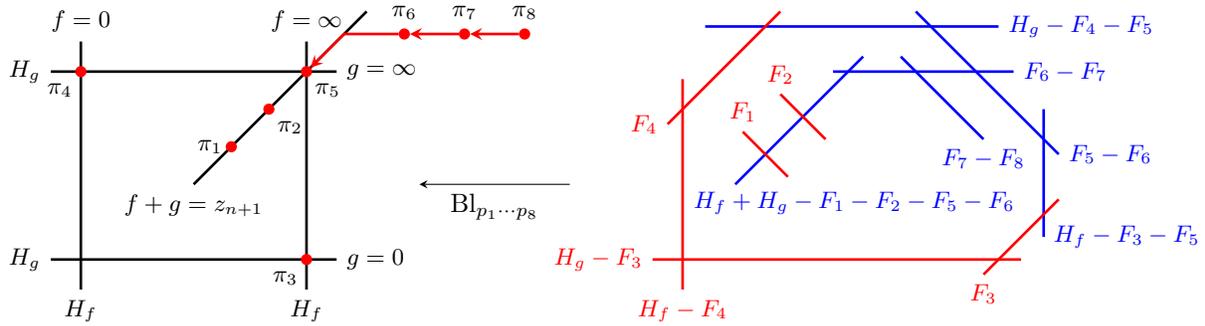
\begin{figure}[ht]
	\begin{tikzpicture}[>=stealth,basept/.style={circle, draw=red!100, fill=red!100, thick, inner sep=0pt,minimum size=1.2mm}]
	\begin{scope}[xshift=0cm,yshift=0cm]
	\draw [black, line width = 1pt] (-0.4,0) -- (3.4,0)	node [pos=0,left] {\small $H_{g}$} node [pos=1,right] {\small $g=0$};
	\draw [black, line width = 1pt] (-0.4,2.5) -- (3.4,2.5) node [pos=0,left] {\small $H_{g}$} node [pos=1,right] {\small $g=\infty$};
	\draw [black, line width = 1pt] (0,-0.4) -- (0,2.9) node [pos=0,below] {\small $H_{f}$} node [pos=1,above] {\small $f=0$};
	\draw [black, line width = 1pt] (3,-0.4) -- (3,2.9) node [pos=0,below] {\small $H_{f}$} node [pos=1,above] {\small $f=\infty$};
	\draw [black, line width = 1pt] (1.5,1) -- (3.8,3.3) node [pos=0,below] {\small $f + g = z_{n+1}$};	
	\node (p1) at (2,1.5) [basept,label={[xshift = -8pt, yshift=-8pt] \small $\pi_{1}$}] {};
	\node (p2) at (2.5,2) [basept,label={[xshift = 8pt, yshift=-15pt] \small $\pi_{2}$}] {};
	\node (p3) at (3,0) [basept,label={[xshift = -8pt, yshift=-15pt] \small $\pi_{3}$}] {};
	\node (p4) at (0,2.5) [basept,label={[xshift = -8pt, yshift=-15pt] \small $\pi_{4}$}] {};
	\node (p5) at (3,2.5) [basept,label={[xshift = 8pt, yshift=-15pt] \small $\pi_{5}$}] {};
	\node (p6) at (4.3,3) [basept,label={[xshift = 0pt, yshift=0pt] \small $\pi_{6}$}] {};
	\node (p7) at (5.1,3) [basept,label={[xshift = 0pt, yshift=0pt] \small $\pi_{7}$}] {};
	\node (p8) at (5.9,3) [basept,label={[xshift = 0pt, yshift=0pt] \small $\pi_{8}$}] {};
	\draw [red, line width = 1pt, ->] (p6) -- (3.5,3) -- (p5);
	\draw [red, line width = 1pt, ->] (p7) -- (p6);
	\draw [red, line width = 1pt, ->] (p8) -- (p7);
	\end{scope}
	\draw [->] (6.5,1)--(4.5,1) node[pos=0.5, below] {$\operatorname{Bl}_{p_{1}\cdots p_{8}}$};
	\begin{scope}[xshift=8cm,yshift=0cm]
	\draw [red, line width = 1pt] (-0.4,0) -- (4.5,0)	node [pos=0, left] {\small $H_{g} - F_{3}$};
	\draw [red, line width = 1pt] (0,-0.4) -- (0,2.4) node [pos=0, below] {\small $H_{f} - F_{4}$};
	\draw [blue, line width = 1pt] (4.8,0.3) -- (4.8,2) node [pos=0, right] {\small $H_{f} - F_{3} - F_{5}$};	
	\draw [blue, line width = 1pt] (0.3,3.1) -- (4.2,3.1) node [pos=1,right] {\small $H_{g} - F_{4} - F_{5}$};
	\draw [red, line width = 1pt] (4,-0.2) -- (5,0.8) node [pos=0,below] {\small $F_{3}$};
	\draw [red, line width = 1pt] (-0.2,1.8) -- (1.3,3.3) node [pos=0, left] {\small $F_{4}$};
	\draw [blue, line width = 1pt] (3.1,3.3) -- (5,1.4) node [pos=1,right] {\small $F_{5} - F_{6}$};
	\draw [blue, line width = 1pt] (2,2.5) -- (4.4,2.5) node [pos=1,right] {\small $F_{6} - F_{7}$};
	\draw [blue, line width = 1pt] (0.7,1) -- (2.4,2.7) node [pos=0,xshift=45pt, yshift=-7pt] {\small $H_{f} + H_{g} - F_{1} - F_{2} - F_{5} - F_{6}$};
	\draw [red, line width = 1pt] (0.8,1.7) -- (1.4,1.1) node [pos=0,above] {\small $F_{1}$};
	\draw [red, line width = 1pt] (1.3,2.2) -- (1.9,1.6) node [pos=0,above] {\small $F_{2}$};
	\draw [blue, line width = 1pt] (2.9,2.7) -- (4,1.6) node [pos=1,below] {\small $F_{7} - F_{8}$};
	\end{scope}
	\end{tikzpicture}
	\caption{The $D_{5}^{(1)}$ Sakai surface for the d-$\Pain{XXXIV}$ equation}
	\label{fig:dP-34-surface}
\end{figure}

Looking at the configuration of $-2$-curves on Figure~\ref{fig:dP-34-surface} we immediately see that this is indeed a $D_{5}^{(1)}$-surface. In fact, the change of basis
on the Picard lattices given by
\begin{equation}\label{eq:change-basis-dP-34}
	\begin{aligned}
		\mathcal{H}_{f} &= \mathcal{H}_{2} + 2\mathcal{H}_{q} - \mathcal{E}_{3} - \mathcal{E}_{5} - \mathcal{E}_{6} - \mathcal{E}_{7}, &\qquad 
				\mathcal{H}_{q} &= \mathcal{H}_{f} + \mathcal{H}_{g} - \mathcal{F}_{1} - \mathcal{F}_{5}, \\ 
		\mathcal{H}_{g} &= \mathcal{H}_{q} + \mathcal{H}_{p} - \mathcal{E}_{3} - \mathcal{E}_{5}, &\qquad 
				\mathcal{H}_{p} &= \mathcal{H}_{f} + 2\mathcal{H}_{g} - \mathcal{F}_{1} - \mathcal{F}_{3}  - \mathcal{F}_{5} - \mathcal{F}_{6}, \\ 
		\mathcal{F}_{1} &= \mathcal{H}_{q} + \mathcal{H}_{p} - \mathcal{E}_{3} - \mathcal{E}_{5} - \mathcal{E}_{6}, &\qquad 
				\mathcal{E}_{1} &= \mathcal{F}_{7},\\
		\mathcal{F}_{2} &= \mathcal{E}_{4}, &\qquad 
				\mathcal{E}_{2} &= \mathcal{F}_{8},\\
		\mathcal{F}_{3} &= \mathcal{H}_{q} - \mathcal{E}_{5}, &\qquad 
				\mathcal{E}_{3} &=  \mathcal{H}_{f} + \mathcal{H}_{g} - \mathcal{F}_{1} - \mathcal{F}_{5} - \mathcal{F}_{6},\\
		\mathcal{F}_{4} &= \mathcal{E}_{8}, &\qquad 
				\mathcal{E}_{4} &= \mathcal{F}_{2},\\
		\mathcal{F}_{5} &= \mathcal{H}_{q} + \mathcal{H}_{p} - \mathcal{E}_{3} - \mathcal{E}_{5} - \mathcal{E}_{7}, &\qquad 
				\mathcal{E}_{5} &= \mathcal{H}_{f} + \mathcal{H}_{g} - \mathcal{F}_{1} - \mathcal{F}_{3} - \mathcal{F}_{5},\\
		\mathcal{F}_{6} &= \mathcal{H}_{q}  - \mathcal{E}_{3}, &\qquad 
				\mathcal{E}_{6} &= \mathcal{H}_{g} - \mathcal{F}_{1},\\
		\mathcal{F}_{7} &= \mathcal{E}_{1}, &\qquad 
				\mathcal{E}_{7} &=\mathcal{H}_{g} -  \mathcal{F}_{5},\\
		\mathcal{F}_{8} &= \mathcal{E}_{2}, &\qquad 
				\mathcal{E}_{8} &= \mathcal{F}_{4}
	\end{aligned}
\end{equation}
results in the following matching of the surface roots,
\begin{alignat*}{2}
	\delta_{0} &= \mathcal{E}_{1} - \mathcal{E}_{2} = \mathcal{F}_{7} - \mathcal{F}_{8},&\qquad  
		\delta_{3} &=\mathcal{H}_{p}-\mathcal{E}_{1} - \mathcal{E}_{3} = \mathcal{F}_{5} - \mathcal{F}_{6}, \\
	\delta_{1} &=  \mathcal{E}_{3} - \mathcal{E}_{4} = \mathcal{H}_{f} + \mathcal{H}_{g} - \mathcal{F}_{1} - \mathcal{F}_{2} - \mathcal{F}_{5} - \mathcal{F}_{6}, &\qquad  
		\delta_{4} &= \mathcal{E}_{5} - \mathcal{E}_{6} = \mathcal{H}_{f} - \mathcal{F}_{3}-\mathcal{F}_{5},\\
	\delta_{2} &= \mathcal{H}_{q}- \mathcal{E}_{1} - \mathcal{E}_{3} = \mathcal{F}_{6} - \mathcal{F}_{7}, &\qquad  
		\delta_{5} &=  \mathcal{E}_{7} - \mathcal{E}_{8} = 	\mathcal{H}_{g} - \mathcal{F}_{4} - \mathcal{F}_{5},
\end{alignat*}
and the symmetry roots,
\begin{equation}\label{eq:sym-roots-dP-34}
	\begin{aligned}
	\alpha_{0} &= \mathcal{H}_{p} - \mathcal{E}_{1} - \mathcal{E}_{2} 
				= \mathcal{H}_{f} + 2\mathcal{H}_{g} - \mathcal{F}_{1} - \mathcal{F}_{3}  - \mathcal{F}_{5} -  \mathcal{F}_{6} - \mathcal{F}_{7} -  \mathcal{F}_{8},\\
	\alpha_{1} &=  \mathcal{H}_{q} - \mathcal{E}_{5} - \mathcal{E}_{6} = - \mathcal{H}_{g} + \mathcal{F}_{1} + \mathcal{F}_{3}, \\
	\alpha_{2} &=\mathcal{H}_{p}-\mathcal{E}_{3} - \mathcal{E}_{4} = \mathcal{H}_{g} - \mathcal{F}_{2} - \mathcal{F}_{3}, \\
	\alpha_{3} &= \mathcal{H}_{q} - \mathcal{E}_{7} - \mathcal{E}_{8} = \mathcal{H}_{f} - \mathcal{F}_{1}-\mathcal{F}_{4}.
	\end{aligned}
\end{equation}

The symplectic form for this point configuration is 
\begin{equation}\label{eq:symp-form-dP-34}
	\omega = \frac{dg \wedge df}{C_{2}(f + g - z_{n+1})},
\end{equation}
and the root variables are 
\begin{equation}\label{eq:root-vars-dP-34}
	a_{0} = \frac{z_{n+2} + C_{3}}{2 C_{2}},\quad a_{1} = -\frac{z_{n+1} + C_{3}}{2 C_{2}},\quad a_{2} = \frac{z_{n+2} - C_{3}}{2 C_{2}},\quad a_{3} = -\frac{z_{n+1} - C_{3}}{2 C_{2}},
\end{equation}
and we see again that the constraint \eqref{eq:dp2-par-constr} holds, $a_{0} + a_{1} = a_{2} + a_{3} = \frac{1}{2}$.

The mapping \eqref{eq:dP-34-map} induces the mapping $\phi_{*}$ on the Picard lattice given by
\begin{alignat*}{2}
	\mathcal{H}_{f}& \mapsto 3 \mathcal{H}_{f} +  \mathcal{H}_{g} - {\mathcal{F}}_{1} - 
	 {\mathcal{F}}_{2} - {\mathcal{F}}_{4} - {\mathcal{F}}_{5} - {\mathcal{F}}_{6} - {\mathcal{F}}_{7}, \qquad&
		\mathcal{H}_{g}& \mapsto \mathcal{H}_{f},  \\
	\mathcal{F}_{1} &\mapsto \mathcal{H}_{f} - \mathcal{F}_{2}, \qquad&
		\mathcal{F}_{5} &\mapsto \mathcal{H}_{f} - \mathcal{F}_{7}.\\
	\mathcal{F}_{2} &\mapsto \mathcal{H}_{f} - {\mathcal{F}}_{1},\qquad&
		\mathcal{F}_{6} &\mapsto \mathcal{H}_{f} - \mathcal{F}_{6},\\
	\mathcal{F}_{3} &\mapsto \mathcal{H}_{f}  - \mathcal{F}_{4},\qquad&
		\mathcal{F}_{7} &\mapsto \mathcal{H}_{f} - \mathcal{F}_{5}, \\
	\mathcal{F}_{4} &\mapsto \mathcal{F}_{8}, \qquad&
		\mathcal{F}_{8} &\mapsto \mathcal{F}_{3}.
\end{alignat*}
Hence we get exactly the same action on the symmetry roots and the root variables as in 
 \eqref{eq:dP-a-root-evol} and \eqref{eq:dP-a-pars-evol},
\begin{align*}
\varphi_{*} : \left(  \alpha_{0}, \alpha_{1}, \alpha_{2}, \alpha_{3} \right) &\mapsto  
			\left( -\alpha_{1}, \alpha_{1} + \alpha_{2} + \alpha_{3}, -\alpha_{3}, 
			\alpha_{0} + \alpha_{1} + \alpha_{3} \right) 
            =\left( -\alpha_{1},  \delta - \alpha_{0}, -\alpha_{3}, 
			\delta - \alpha_{2}\right)\\
			\overline{a}_{0} &= 1 - a_{1},\qquad \overline{a}_{1} = -a_{0},\qquad \overline{a}_{2} = 1 - a_{3}, \qquad \overline{a}_{3} = - a_{2}.			
\end{align*}
In fact, d-$\Pain{II}$ discrete Painlev\'e equation \eqref{eq:dP-II-eq} and d-$\Pain{XXXIV}$ equation \eqref{eq:CJ-dP34} are related by the following
birational change of variables and parameter identification:
\begin{equation}\label{eq:dP2to32-coords}
 \left\{\begin{aligned}
 	x(f,g)&=\frac{g - f - C_{3}}{f + g - z_{n+1}},\\
	y(f,g)&= 1 - \frac{g(f + g - z_{n+1})}{2g - z_{n+1} - C_{3}},\\
	\alpha &= C_{2},\ \beta= C_{1} + C_{2}, \gamma = - C_{3},
 \end{aligned}\right.
\qquad 
	\left\{\begin{aligned}
 	f(x,y)&=(x-1)(y-1) + \frac{(n \alpha + \beta)x + \gamma}{x+1},\\[4pt]
	g(x,y)&=-(x+1)(y-1),\\[6pt]
	C_{1} &= \beta - \alpha,\ C_{2} = \alpha,\ C_{3} = - \gamma.  
	\end{aligned}\right.
\end{equation}	

However, for the Freud weight, parameters $C_{i}$ take very special values \eqref{eq:MW-pars} and the corresponding root variables become
\begin{equation}\label{eq:root-vars-dP-34-sp}
	a_{0} = 1 + \frac{n}{2},\quad a_{1} = -\frac{n+1}{2},\quad a_{2} = \frac{n+1}{2},\quad a_{3} = -\frac{n}{2},
\end{equation}
Note that we now have $a_{1} + a_{2} = 0$, which is the nodal curve condition for the symmetry root $\alpha_{1} + \alpha_{2} = \mathcal{F}_{1} - \mathcal{F}_{2}$.
Indeed, with these values of parameters the base points $\pi_{1}$ and $\pi_{2}$ coalesce along the line $f + g = \frac{4n + 2}{a^{2}}$ 
into the point $\left(\frac{2n}{a^{2}},\frac{2(n+1)}{a^{2}}\right)$. A more careful computation, in fact, shows that we get a cascade of two infinitely close 
points (that we still denote by $\pi_{1,2}$):
\begin{equation}\label{eq:pts12-cascade}
	\pi_{1}\left(\frac{2n}{a^{2}},\frac{2(n+1)}{a^{2}}\right) \leftarrow \pi_{2}(u_{1} = 0, v_{1} = -1).
\end{equation}
The resulting Sakai surface is shown on Figure~\ref{fig:dP-34-surface-sp}, the nodal curve $F_{1} - F_{2}$ is a $-2$-curve disjoint from the 
anti-canonical divisor $-K_{\mathcal{X}}$. 

\begin{figure}[ht]
	\begin{tikzpicture}[>=stealth,basept/.style={circle, draw=red!100, fill=red!100, thick, inner sep=0pt,minimum size=1.2mm}]
	\begin{scope}[xshift=0cm,yshift=0cm]
	\draw [black, line width = 1pt] (-0.4,0) -- (3.4,0)	node [pos=0,left] {\small $H_{g}$} node [pos=1,right] {\small $g=0$};
	\draw [black, line width = 1pt] (-0.4,2.5) -- (3.4,2.5) node [pos=0,left] {\small $H_{g}$} node [pos=1,right] {\small $g=\infty$};
	\draw [black, line width = 1pt] (0,-0.4) -- (0,2.9) node [pos=0,below] {\small $H_{f}$} node [pos=1,above] {\small $f=0$};
	\draw [black, line width = 1pt] (3,-0.4) -- (3,2.9) node [pos=0,below] {\small $H_{f}$} node [pos=1,above] {\small $f=\infty$};
	\draw [black, line width = 1pt] (1.5,1) -- (3.8,3.3) node [pos=0,below] {\small $f + g = z_{n+1}$};	
	\node (p1) at (1.7,1.2) [basept,label={[xshift = -8pt, yshift=-8pt] \small $\pi_{1}$}] {};
	\node (p2) at (2.5,1.5) [basept,label={[xshift = 0pt, yshift=-15pt] \small $\pi_{2}$}] {};
	\node (p3) at (3,0) [basept,label={[xshift = -8pt, yshift=-15pt] \small $\pi_{3}$}] {};
	\node (p4) at (0,2.5) [basept,label={[xshift = -8pt, yshift=-15pt] \small $\pi_{4}$}] {};
	\node (p5) at (3,2.5) [basept,label={[xshift = 8pt, yshift=-15pt] \small $\pi_{5}$}] {};
	\node (p6) at (4.3,3) [basept,label={[xshift = 0pt, yshift=0pt] \small $\pi_{6}$}] {};
	\node (p7) at (5.1,3) [basept,label={[xshift = 0pt, yshift=0pt] \small $\pi_{7}$}] {};
	\node (p8) at (5.9,3) [basept,label={[xshift = 0pt, yshift=0pt] \small $\pi_{8}$}] {};
	\draw [red, line width = 1pt, ->] (p2) -- (2,1.5) -- (p1);
	\draw [red, line width = 1pt, ->] (p6) -- (3.5,3) -- (p5);
	\draw [red, line width = 1pt, ->] (p7) -- (p6);
	\draw [red, line width = 1pt, ->] (p8) -- (p7);
	\end{scope}
	\draw [->] (6.5,1)--(4.5,1) node[pos=0.5, below] {$\operatorname{Bl}_{p_{1}\cdots p_{8}}$};
	\begin{scope}[xshift=8cm,yshift=0cm]
	\draw [red, line width = 1pt] (-0.4,0) -- (4.5,0)	node [pos=0, left] {\small $H_{g} - F_{3}$};
	\draw [red, line width = 1pt] (0,-0.4) -- (0,2.4) node [pos=0, below] {\small $H_{f} - F_{4}$};
	\draw [blue, line width = 1pt] (4.8,0.3) -- (4.8,2) node [pos=0, right] {\small $H_{f} - F_{3} - F_{5}$};	
	\draw [blue, line width = 1pt] (0.3,3.1) -- (4.2,3.1) node [pos=1,right] {\small $H_{g} - F_{4} - F_{5}$};
	\draw [red, line width = 1pt] (4,-0.2) -- (5,0.8) node [pos=0,below] {\small $F_{3}$};
	\draw [red, line width = 1pt] (-0.2,1.8) -- (1.3,3.3) node [pos=0, left] {\small $F_{4}$};
	\draw [blue, line width = 1pt] (3.1,3.3) -- (5,1.4) node [pos=1,right] {\small $F_{5} - F_{6}$};
	\draw [blue, line width = 1pt] (2,2.5) -- (4.4,2.5) node [pos=1,right] {\small $F_{6} - F_{7}$};
	\draw [blue, line width = 1pt] (0.7,1) -- (2.4,2.7) node [pos=0,xshift=45pt, yshift=-7pt] {\small $H_{f} + H_{g} - F_{1} - F_{2} - F_{5} - F_{6}$};
	\draw [blue, line width = 1pt] (1.4,1) -- (2,1.6) node [pos=1,xshift=10pt, yshift=6pt] {\small $F_{1} - F_{2}$};
	\draw [red, line width = 1pt] (1,1.9) -- (1.9,1) node [pos=0,above] {\small $F_{2}$};
	\draw [blue, line width = 1pt] (2.9,2.7) -- (4,1.6) node [pos=1,below] {\small $F_{7} - F_{8}$};
	\end{scope}
	\end{tikzpicture}
	\caption{The $D_{5}^{(1)}$ Sakai surface for the constrained d-$\Pain{XXXIV}$ equation}
	\label{fig:dP-34-surface-sp}
\end{figure}
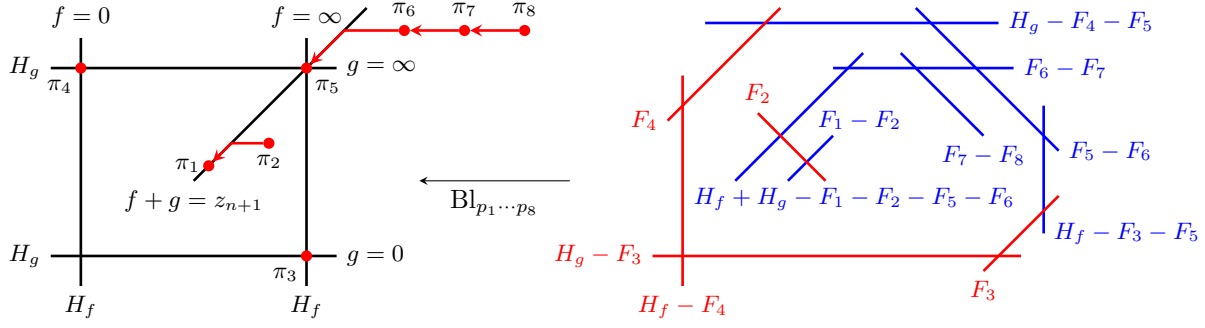

Thus, now in addition to fixing the constraint \eqref{eq:dp2-par-constr} we now also need to fix the nodal curve condition $a_{1} + a_{2} = 0$,
which further restricts the symmetry group of the equation.

\begin{theorem} 
	The subgroup of $\widehat{W}(A_3^{(1)})$ compatible with the constraint 
	\begin{equation} \label{eq:constraintdP34}
		a_{0} + a_{1} = a_{2} + a_{3} = \frac{1}{2},\qquad a_{1} + a_{2} = 0
	\end{equation}
	on root variables is 
    \begin{equation}
        H = \left\{ w \in G ~|~ w(\alpha_1+\alpha_2)=\alpha_1+\alpha_2\right\} 
        = \left\langle r_{\gamma_0},r_{\gamma_1},\sigma_1\right\rangle \cong \widetilde{W}\left(A_1^{(1)}\right). 
    \end{equation}
\end{theorem}
\begin{proof}
    We already have from Theorem \ref{th:maintheorem} that the subgroup of $\widehat{W}\left(A_3^{(1)}\right)$ compatible with the constraint $a_{0} + a_{1} = a_{2} + a_{3} = \frac{1}{2}$ is $G=\operatorname{Stab}\left\{\alpha_0+\alpha_1,\alpha_2+\alpha_3\right\}$.
    To find the additional condition on the actions of symmetries on $\operatorname{Pic(\mathcal{X})}$ necessary to respect the additional constraint $a_1+a_2=0$, we require geometric considerations. (The need for these considerations arises from the fact that the additional constraint corresponds to the value of the period map on a root of the $A_3^{(1)}$ system being 0, so implies the existence of the nodal curve (see \cite[Prop. 22]{Sak:2001:RSAWARSGPE}) whereas
    the constraint $a_{0} + a_{1} = a_{2} + a_{3} = \frac{1}{2}$ does not correspond to any nodal curves in this way).
    
    While on the level of root variables the actions of some elements may respect the constraint $a_1+a_2=0$ algebraically, on the level of $\operatorname{Pic(\mathcal{X})}$ they might not respect effectiveness of divisor classes and, hence, do not give automorphisms of the sub-family of surfaces defined by the existence of this nodal curve. 
    For example, the action of the reflection $r_{\beta_1}= r_{\alpha_1+\alpha_2} \in G$ on root variables is such that $a_1+a_2\mapsto -a_1-a_2$, so it preserves the subset $\{a_1+a_2=0\}$, but on the level of $\operatorname{Pic}(\mathcal{X})$ it does not preserve effectiveness of divisor classes on the surfaces with nodal curves, so its Cremona action does not restrict to the sub-family. 
    A generic surface in the sub-family defined by the constraints \eqref{eq:constraintdP34} has only a single nodal curve corresponding to the class $\alpha_1+\alpha_2\in\operatorname{Pic}(\mathcal{X})$. 
    Symmetries compatible with these constraints must preserve the subset of $\operatorname{Pic}(\mathcal{X})$ corresponding to the set of nodal curves, so we arrive at the description of the symmetry subgroup $H$ of the sub-family of surfaces defined by the constraints \eqref{eq:constraintdP34} as the elements of $G$ that fix $\alpha_1+\alpha_2$.

    From the description of $G = \left\langle r_{\beta_0},r_{\beta_1},\sigma_1\sigma_2\sigma_1\sigma_2\right\rangle \times \left\langle r_{\gamma_0},r_{\gamma_1},\sigma_1\right\rangle$, we see that the only element of the first factor $\left\langle r_{\beta_0},r_{\beta_1},\sigma_1\sigma_2\sigma_1\sigma_2\right\rangle$ that fixes $\alpha_1+
    \alpha_2$ is the identity. 
    To see this, note that $\left\langle r_{\beta_0},r_{\beta_1},\sigma_1\sigma_2\sigma_1\sigma_2\right\rangle=\left\langle r_{\beta_1}\right\rangle \ltimes \left\langle \sigma_1\sigma_2\sigma_1\sigma_2 r_{\beta_1}\right\rangle\cong W(A_1)\ltimes \mathbf{T}_{P(A_1)}$, where $\sigma_1\sigma_2\sigma_1\sigma_2 r_{\beta_1} : \left( \alpha_0,\alpha_1,\alpha_2,\alpha_3 \right) \mapsto \left(\delta-\alpha_1,-\alpha_0,-\alpha_3,\delta-\alpha_2\right)$ is the quasi-translation in $\widehat{W}(A_3^{(1)})$ corresponding to a translation in $\widetilde{W}(A_1^{(1)})$ which generates the subgroup of translations associated to the weight lattice $P(A_1)$ of the underlying $A_1$ root system.
    On the other hand, every element of $\left\langle r_{\gamma_0},r_{\gamma_1},\sigma_1\right\rangle$ fixes $\alpha_1+\alpha_2$, so we have the result.
    
\end{proof}


\section{Discussion}

In this paper we advocate the point of
view that what should be called the {\em symmetry group} of a discrete Painlev\'e equation is the group of symmetries of the surface (sub-)family forming its configuration space, and the translation elements of which generate the resulting dynamics. In particular, it is an extended affine Weyl group whose birational representation on the surface sub-family generates the equation. 
For many discrete Painlev\'e equations that appear in the literature, this symmetry group will not be the  full, generic, symmetry group attached to the surface type of their configuration space as given in Sakai's classification, but rather a subgroup of 
that generic symmetry group.

We explain this crucial difference on a specific example of a well-known d-$\Pain{II}$ equation that does not correspond to a genuine translation element of the generic symmetry group.
We identify the surface sub-family that is left invariant under the dynamics defined by the equation and then calculate its symmetry group in two different ways: first by a brute force calculation and then using elements from the normalizer theory of parabolic subgroups of Coxeter groups. It is important to emphasize here that the first approach, the brute force calculation, will become unwieldy and ultimately unfeasible for equations with higher dimensional (symmetry) root lattices. The normalizer theory-based proof, on the contrary, is almost entirely algorithmic and further provides insight into the nature of quasi-translations in the generic symmetry group as translations in the relevant symmetry subgroup.

We end the paper with a similar analysis for a sub-case of our main example, which arises in the study of gap probabilities for Freud unitary ensembles, and the symmetry group of which is even further restricted due to the appearance of a nodal curve on the surface on which the equation is regularized.

\medskip

\textbf{Acknowledgments:}\quad
RW would like to thank the Japan Society for the Promotion of Science (JSPS) for financial support through the KAKENHI grants 22H01130 and 23K22401.
AS would like to thank JSPS for financial support through the KAKENHI grant 24K22843. 
AD would like to acknowledge support from BIMSA start-up fund. YS would like to thank R. B. Howlett and N. Nakazono
for the enlightening discussions.

\providecommand{\bysame}{\leavevmode\hbox to3em{\hrulefill}\thinspace}
\providecommand{\MR}{\relax\ifhmode\unskip\space\fi MR }
\providecommand{\MRhref}[2]{%
  \href{http://www.ams.org/mathscinet-getitem?mr=#1}{#2}
}
\providecommand{\href}[2]{#2}

\end{document}